\documentclass[10pt,conference]{IEEEtran}
\input{psfig.sty}
\input{epsf.sty}
\usepackage{hyperref} 
\usepackage{fancybox}   
\usepackage{psfig}      
\usepackage{graphicx}
\usepackage{amsmath}
\usepackage{color}
\usepackage{epsfig}
\usepackage{amssymb}
\usepackage{verbatim}
\usepackage{enumitem}
\usepackage{bbm}

\IEEEoverridecommandlockouts

\newcommand{\remove}[1]{}

\newtheorem{thm}{Theorem}

\newtheorem{lem}{Lemma}

\newtheorem{rem}{Remark}

\usepackage{setspace}
\usepackage{subfigure}

\begin{document}

\title{Measurement Based Impromptu Deployment of a Multi-Hop Wireless Relay Network \thanks{This 
work was supported by the Indo-French Centre for the Promotion of Advanced Research (Project 4900-ITB), the 
Department of Electronics and Information Technology, via the Indo-US PC3 project, and the Department of Science and 
Technology (DST), via the J.C. Bose Fellowship.}}


\newcounter{one}
\setcounter{one}{1}
\newcounter{two}
\setcounter{two}{2}



\author{
Arpan~Chattopadhyay$^\fnsymbol{one}$, Marceau~Coupechoux$^\fnsymbol{two}$, and Anurag~Kumar$^\fnsymbol{one}$\\
\parbox{0.49\textwidth}{\centering $^\fnsymbol{one}$Dept. of ECE, Indian Institute of Science\\
Bangalore 560012, India\\
arpanc.ju@gmail.com, anurag@ece.iisc.ernet.in}
\hfill
\parbox{0.49\textwidth}{\centering $^\fnsymbol{two}$Telecom ParisTech and CNRS LTCI \\
Dept. Informatique et R\'eseaux\\
23, avenue d'Italie, 75013 Paris, France\\
marceau.coupechoux@telecom-paristech.fr}
}

\maketitle
\thispagestyle{empty}


\begin{abstract}
We study the problem of optimal sequential (``as-you-go") deployment
  of wireless relay nodes, as a person walks along a line of random
  length (with a known distribution).  The objective is to create an
  impromptu multihop wireless network for connecting a packet source
  to be placed at the end of the line with a sink node located at the
  starting point, to operate in the light traffic regime.  In walking
  from the sink towards the source, at every step, measurements yield
  the transmit powers required to establish links to one or more
  previously placed nodes. Based on these measurements, at every step,
  a decision is made to place a relay node, the overall system
  objective being to minimize a linear combination of the expected
  sum power (or the expected maximum power) required to deliver a
  packet from the source to the sink node and the expected number of
  relay nodes deployed. For each of these two objectives, two
  different relay selection strategies are considered: (i) each relay
  communicates with the sink via its immediate previous relay, (ii)
  the communication path can skip some of the deployed relays.  With
  appropriate modeling assumptions, we formulate each of these
  problems as a Markov decision process (MDP).  We provide the optimal
  policy structures for all these cases, and provide illustrations of
  the policies and their performance, via numerical results, for some
  typical parameters.
\end{abstract}

\vspace{-1 mm}

\section{Introduction}
\label{Introduction}

\vspace{-1 mm}

Wireless interconnection of resource-constrained mobile user devices or wireless sensors to
the wireline infrastructure via relay nodes is an important
requirement, since a direct one-hop link from the source node to
the infrastructure ``base-station" may not always be feasible, due to distance or poor channel condition. 
The relays could be battery operated
radio routers or other wireless sensors in the wireless sensor network context, or other users' devices 
in the cellular context. The relays are also resource constrained and a cost might be
involved in engaging or placing them. Hence, there arises the
problem of {\em optimal relay placement}.

Motivated by the above larger problem, we consider the problem of ``as-you-go" deployment of relay nodes
along a line, between a sink node and a source node (see
Figure~\ref{fig:line-network-general}), where the deployment operative
starts from the sink node, places relay nodes along the line, and
places the source node where the line ends.  The problem is motivated
by the need for impromptu deployment of wireless networks by ``first
responders," for situation monitoring in an emergency such as a
building fire or a terrorist siege. Such problems can also arise when
deploying wireless sensor networks in large difficult terrains (such
as forests) where it is difficult to plan a deployment due to the
unavailability of a precise map of the terrain, or when such networks
need to be deployed and redeployed quickly and there is little time in
between to plan, or in situations where the deployment needs to be
stealthy (for example, when deploying sensor networks for detecting
poachers or fugitives in a forest). 

\begin{figure}[!t]
\centering
\includegraphics[scale=0.32]{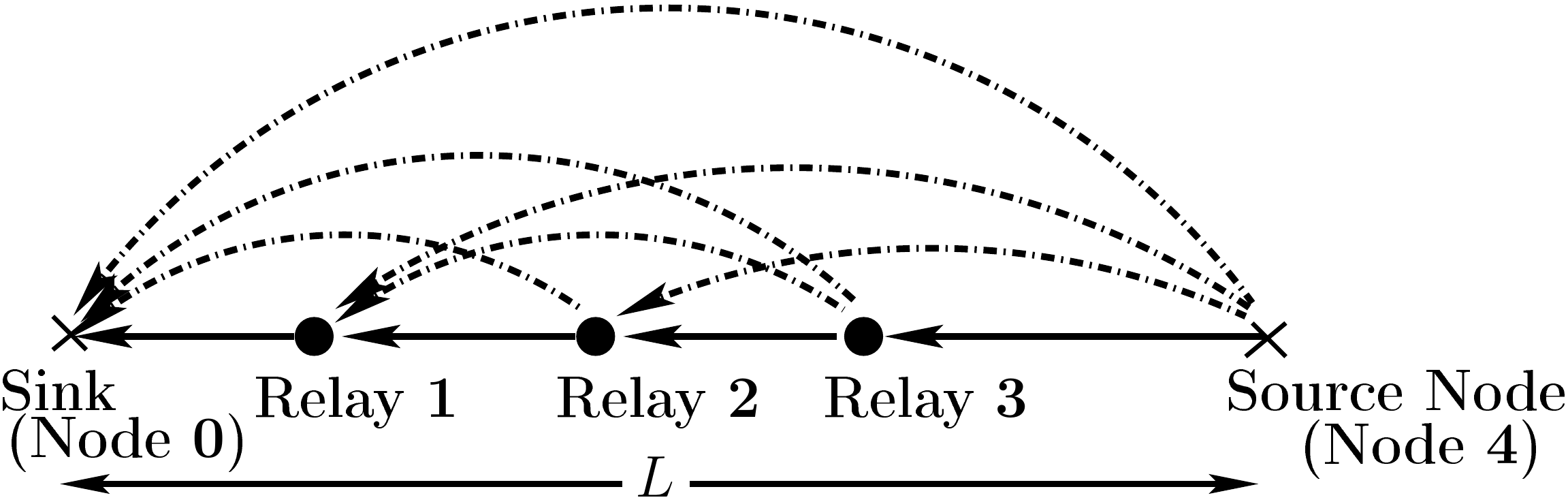}
\vspace{-2mm}
\caption{A line network with one source, one sink and three relays.}
\label{fig:line-network-general}
\vspace{-10mm}
\end{figure}

In this paper, we consider the problem of
as-you-go deployment along a line of unknown random length, $L$, whose
distribution is known.  The transmit power required to establish a
link (of a certain minimum quality) between any two nodes is modeled
by a random variable capturing the effect of path-loss and
shadowing. There is a cost for placing a relay, and the communication
cost of a deployment is measured as some function of the powers
required to communicate over the links. We consider two performance
measures: the sum-power and the max-power along the path from the
source to the sink, under two different relay selection strategies:
(i) each relay communicates with the sink via its immediate previous
relay, (ii) the communication path can skip some of the deployed
relays. Under certain assumptions on the distribution of $L$, and the
powers required at the relays, we formulate each of the sequential
placement problems as a total cost Markov decision process (MDP).

The optimal policies for various MDPs formulated in this paper turn
out to be {\em threshold} policies; the decision to place a relay at a
given location involves the power required to establish a link to one
or more previous nodes, and the distance to one or more previous nodes
(depending on the objective and the relay selection strategy). Our
analysis and numerical work also suggest that allowing the possibility
of skipping some of the deployed relays may result in a reduction in
the total cost.

\vspace{-3 mm}
  
\subsection{Related Work}

\vspace{-1 mm}

There has been increasing interest 
in the research community to explore the impromptu relay placement problem in recent years.  
Howard et al., in \cite{howard-etal02incremental-self-deployment-algorithm-mobile-sensor-network},
provide heuristic algorithms for incremental deployment of
sensors with the objective of
covering the deployment area. Souryal et al., in \cite{souryal-etal07real-time-deployment-range-extension}, address the
problem of impromptu deployment of static wireless networks with an extensive study of indoor RF link quality variation.
The work reported in \cite{aurisch-tlle09relay-placement-emergency-response} use similar approach for relay deployment. 
Recently, Liu et al. (\cite{liu-etal10breadcrumb}) describe a {\em breadcrumbs} system for aiding firefighters inside buildings. 
However, there has been little effort to rigorously formulate the problem in order to derive optimal policies, 
from which insights can be gained, and which can be compared in performance to reasonable heuristics. 
Recently, Sinha et al. (\cite{sinha-etal12optimal-sequential-relay-placement-random-lattice-path}) 
have provided an MDP formulation for establishing a multi-hop 
network between a destination and an unknown source location
by placing relay nodes along a random lattice path. They assume a given deterministic mapping between power and wireless 
link length, and, hence, do not consider the statistical variability (due to shadowing) of the transmit power 
required to maintain the link quality over links of a given length. We, however, 
consider such variability, and therefore bring in the idea of measurement based impromptu placement.

\vspace{-2 mm}

\subsection{Organization}

 In Section \ref{sec:system_model_and_notation}, the system model and the basic notation used in this work are discussed.

In Section \ref{sec:only_adjacent_nodes}, the problem of sequential relay placement is addressed, under the assumption 
that a packet originating from the source makes a hop-by-hop traversal through all 
relay nodes. We formulate the problems with sum power and max-power objectives 
as MDPs and establish the optimal policy structures analytically. We show that that, in each case, the decision to place or not to place at the current position depends on a comparison of the transmit 
power for establishing a link from the current position, with a threshold that depends on the state of the Markov decision process.

In Section \ref{sec:not_only_adjacent_nodes}, we again address the same problems as in 
Section \ref{sec:only_adjacent_nodes}, but we relax the restriction that the links of the path from the 
source to the sink must be between adjacent deployed relays. 
This relaxation leads us to the formulation of MDPs with a more complicated state space. The optimal 
policies again turn out to be threshold policies. We provide numerical examples, 
using parameters similar to those that occur in commercially available wireless sensor networking devices. 
The performance improvement due to skipping relays is demostrated.

 We conclude in Section \ref{sec:conclusion}.

\section{System Model and Notation}\label{sec:system_model_and_notation}
\vspace{-2 mm}

\subsection{Length of the Line}\label{subsection:length_of_the_line}
The length $L$ of the line is a priori unknown, but there is prior 
information (e.g., the mean length $\overline{L}$) that leads us to model $L$ as a geometrically 
distributed number of steps. The step length $\delta$ 
(whose values will typically be several meters, e.g., $2$~meters in our numerical work) 
and the mean length of the line, $\overline{L}$, can be used 
to obtain the parameter of the geometric distribution, i.e., the probability $\theta$ that the line ends at the 
next step. In the problem formulation, we assume $\delta=1$ for simplicity.\footnote{The geometric distribution 
is the maximum entropy discrete probability 
mass function with a given mean. Thus, by using the geometric distribution, 
we are leaving the length of the line as uncertain as we can, given the prior knowledge of $\overline{L}$.} 
All distances are assumed to be integer multiples of $\delta$.


\subsection{Deployment Process and Some Notation}\label{subsection:deployment_process_notation}
\vspace{-1mm}

As the person walks along the line, at each step he measures the link quality from the 
current location to one or more than one previous node and accordingly decides whether to place a relay at the current 
location or not. After the deployment process is complete (at the end of the line where the source is placed), we denote the 
number of deployed relays by $N$, which is a random number, with the randomness coming 
from the randomness in the link qualities and in the length of the line. 
As shown in Figure~\ref{fig:line-network-general}, the sink is called Node $0$, the relay 
closest to the sink is called Node $1$, and finally the source is called Node $(N+1)$. 
The link whose transmitter is Node $i$ and receiver is Node 
$j$ is called link $(i,j)$. A generic link is denoted by $e$.

\vspace{-3.5 mm}
\subsection{Traffic Model}\label{subsection:traffic_model}
\vspace{-1.5 mm}
We consider a traffic model where the traffic is so
low that there is only one packet in the network at a time; we call
this the ``lone packet model.''  As a consequence of this assumption,
(i) the transmit power required over each link depends only on the
propagation loss over that link, as there are no simultaneous
transmissions to cause interference, and (ii) the transmission delay
over each link is easily calculated, as there are no simultaneous
transmitters to contend with. This permits us to easily write down the
communication cost on a path over the deployed relays.

It was shown in \cite{bhattacharya-kumar12qos-aware-survivable-network-design} that a network operating under 
CSMA/CA medium access, and designed for carrying any positive traffic, with some QoS (in terms of the 
packet delivery probability), must  necessarily be able 
to provide the desired QoS to lone packet traffic.   
As-you-go deployment of wireless networks 
while meeting QoS objectives for specific positive packet arrival rates is a topic of our ongoing research.

\subsection{Channel Model}\label{subsection:channel_model}

For our network performance objective (see Section~\ref{subsection:objective}), we need the transmit power required to sustain a 
certain quality of communication over a link.  In order to model this required power, we consider the usual 
aspects of path-loss, shadowing, and fading. A link is considered to be in \emph{outage} if  the received 
signal strength drops below $P_{rcv-min}$ (due to fading) (e.g., below -88~dBm). The transmit power that 
we use is such that the probability of outage is less than a small value (say 5\%).   For a generic link 
of length $r$, we denote by $\Gamma_{r}$ the transmit power required; due to shadowing, this is modeled 
as a random variable.  Since practical radios can only be set to transmit at a finite set of power levels, 
the random variable $\Gamma_{r}$ takes values from a discrete set, $\mathcal{S}$. The distribution function 
of $\Gamma_r$ is denoted by  $G_r(\cdot)$, and the probability mass function (p.m.f.) by $g(r, \cdot)$, i.e., 
$g(r,\gamma):=P(\Gamma_r=\gamma)$ for all $ \gamma\in \mathcal{S}$; $g(r,\gamma)$ is the probability 
that at least the transmit power level $\gamma$ is required to establish a link of length $r$. 
Since the required transmit power increases with distance, 
we assume that $\{G_r\}_{r=1,2,...}$ is a sequence of distributions {\em stochastically increasing} 
(for definition, see~\cite{shaked-shanthikumar06stochastic-orders}) in $r$. 
We also need to talk about a specific link, say $e$; 
we will denote the transmit power required for this link by $\Gamma^{(e)}$.
We assume that the powers required to establish any two different links in the network 
are independent, i.e., $\Gamma^{(e_1)}$ is independent of $\Gamma^{(e_2)}$ for $e_1 \neq e_2$. 
Spatially correlated shadowing will be
considered in our future work.

\subsection{Deployment Objective}\label{subsection:objective}
In this paper, we do not consider the possibility of
another person following behind, who can learn from the measurements
and actions of the first person, thereby supplementing
the actions of the preceding individual.  Our objective is to design
relay placement policies so as to minimize the sum of the expected sum
power/ maximum power (to deliver a packet from the
source node to the sink node) and the expected cost of placing the
relays (the expected number of relays multiplied by the relay cost,
$\xi$).  By a standard constraint relaxation, this problem also arises
from the problem of minimizing the expected sum/ max
power, subject to a constraint on the mean number of relays.  Such a
constraint can be justified if we consider relay deployment for
multiple source-sink pairs over several different lines of mean length
$\overline{L}$, given a large pool of relays, and we are only
interested in keeping small the total number of relays over all these
deployments.

The max-power objective is a valid one in a typical sensor network setting since each of the battery-operated relays must 
use as little power as possible in order to maximize the network lifetime. The sum-power objective may be useful in a different 
scenario. Consider a mobile station (MS) trying to establish a {\em multihop} connection to 
a base station (BS) at an unknown distance in order to download data, 
and there is a continuum of other 
nodes between them. Each node can be used as a relay only if it is paid a certain price.  
The price may have a fixed component (corresponds to the cost $\xi$ of a relay) and a variable component proportional to the 
power used by the relay to serve the MS. The MS could send out a probe towards the BS; the probe needs to 
sequentially establish a multihop path using other nodes along the path as relays. 
The formulation for this problem will be analogous to that for the sequential relay 
placement problem with the sum-power objective. Also, in the context of global energy saving, 
it is interesting to have relay deployment policies that minimize the sum power.

Note that our problem formulation is applicable to the situation in which a relay can be set to a low 
power state except when it has to receive or transmit. If the relays always keep their
receivers on, with the current drawn from the battery in the
receiving mode being the same as the current required at the transmit
mode, then the battery lifetime will depend on the current
drawn in the receive mode, since for light traffic the node
will transmit rarely. Also note that, our formulation
is capable of using any monotonically increasing function of
the power at a node as the objective to be optimized, rather
than directly using the sum/max power objective. The function
could denote the current requirement for a particular transmit
power, which will in turn govern the lifetime.

\subsection{Routing over the Deployed Relays}\label{subsection:routing_protocol}
After node deployment, the routes could be constrained
so as to allow transmissions only between adjacent nodes, i.e., the
routes use solely the links represented by the solid lines in
Figure~\ref{fig:line-network-general}; we consider this problem in
Section~\ref{sec:only_adjacent_nodes}.
However, after deployment, it may turn out that it is better that the
route from the source to the sink skips some relays (e.g., in
Figure~\ref{fig:line-network-general}, if the channel between the
source node and relay~$2$ is very good, it could be better to directly
transmit from the source node to relay~$2$ without using relay
$3$). Hence, while formulating the problem, it would be beneficial to
permit the possibility that some of the dotted links in
Figure~\ref{fig:line-network-general} can be used after deployment;
this problem is solved in Section~\ref{sec:not_only_adjacent_nodes}. 


\section{Relaying via Adjacent Previous Node Only}
\label{sec:only_adjacent_nodes}

In this section we allow relaying from the source to the sink only by each relay passing 
the packet to the immediate previous relay, in the order of deployment. Thus, this is the measurement-based 
extension to the problem in \cite{mondal-etal12impromptu-deployment_NCC}.

\vspace{1mm}

\subsection{Sum-Power Objective}\label{subsec:sum_power_adjacent}

\vspace{1mm}
\subsubsection{Problem Formulation}\label{subsubsec:formulation_sum_power_adjacent}

Our problem is to place the relay nodes sequentially such that the expected sum of the total power cost and the relay cost 
is minimized. We formulate this problem as an MDP with state $(r,\gamma)$, where $r$ is the distance of the current location 
from the previous node and and $\gamma$ is the transmit power required to establish a link to the previous node from the current 
location. Based on $(r,\gamma)$ a decision is made whether to place a relay at the current position or not. $\mathbf{0}$ denotes 
the state at the beginning of the process (at the sink node). 
When the source is placed, the process terminates and the system 
enters and stays forever at a state $\mathbf{e}$.  
The action space is $\{\textit{place}, \, \textit{do not place}\}$. The randomness comes from the random length $L$ and 
the randomness in $\Gamma^{(e)}$.

The problem we seek to solve is:
\begin{eqnarray}
 \min_{\pi \in \Pi} \mathbb{E}_{\pi} \bigg(\sum_{i=1}^{N+1}\Gamma^{(i,i-1)}+\xi N \bigg)  
\label{eqn:unconstrained_total_power_problem} 
\end{eqnarray}
where $\Pi$ is the set of all stationary deterministic Markov placement policies and $\pi$ is a specific stationary deterministic 
Markov placement policy. Any deterministic Markov policy $\pi$ is a sequence of 
mappings $\{\mu_k\}_{k \geq 1}$, where $\mu_k$ takes the state of the system at time $k$ (the $k$-th step from the sink, in 
the context of our problem)
and maps it into any one of the two actions 
$\{\textit{place}, \, \textit{do not place}\}$. If $\mu_k$ does not depend on $k$, then the policy is called 
\textit{stationary policy}. By proposition $1.1.1$ of \cite{bertsekas07dynamic-programming-optimal-control-2}, 
we can restrict ourselves to the class of randomized Markov policies. 
The justification for restriction to stationary deterministic 
policies will be given later.

Solving (\ref{eqn:unconstrained_total_power_problem}) 
also helps in solving the following constrained problem (see \cite{altman99constrained-mdp}):
\begin{eqnarray}
 \min_{\pi \in \Pi} \mathbb{E}_{\pi} \bigg(\sum_{i=1}^{N+1}\Gamma^{(i,i-1)}  \bigg) \text{ such that } \mathbb{E}_{\pi}N \leq M, 
\label{eqn:constrained_total_power_problem} 
\end{eqnarray}
where $M$ is a constraint on the mean number of relays deployed. 
In this paper, however, we consider only the unconstrained problem.
 
If the state is $(r,\gamma)$ and a relay is placed, the relay cost $\xi$ and the power cost $\gamma$ is incurred 
at that step. We do not count 
the price of the source node, but include the power used by the source in our cost. {\em No cost is incurred if we do not place a 
relay at a certain location.} 
Note that $\mathbf{0}$ also denotes 
the state immediately after placing a relay, 
since the process regenerates whenever a relay is placed (this follows from the memoryless property 
of geometric distribution and the independence of $\Gamma^{(i,j)}$ and $\Gamma^{(k,l)}$ for $(i,j) \neq (k,l)$).  
Let us define $J_{\xi}(r, \gamma)$ and $J_{\xi}(\mathbf{0})$ to be the optimal expected cost-to-go starting from 
state $(r, \gamma)$ and state $\mathbf{0}$ respectively. 
\vspace{-1mm}
\subsubsection{Bellman Equation}\label{subsubsec:bellman_sum_power_adjacent}
Here we have an infinite horizon total cost MDP with a 
countable state space and finite action space. The assumption P of Chapter~$3$ in \cite{bertsekas07dynamic-programming-optimal-control-2} is 
satisfied here, since the single-stage costs are nonnegative.
 Hence, by Proposition $3.1.1$ of \cite{bertsekas07dynamic-programming-optimal-control-2}, 
the optimal value function $J_{\xi}(\cdot)$ satisfies the following Bellman equation:
\begin{eqnarray}
 J_{\xi}(r,\gamma)&=&\min \bigg\{ \overbrace{\xi+ \gamma + J_{\xi}(\mathbf{0})}^{c_p} , \nonumber\\
&& \overbrace{\theta \mathbb{E} (\Gamma_{r+1})+(1-\theta) \mathbb{E} J_{\xi}\left(r+1, \Gamma_{r+1}\right)}^{c_{np}} \bigg\} \nonumber\\
 J_{\xi}(\mathbf{0})&=&\theta \mathbb{E} (\Gamma_{1}) + (1-\theta) \mathbb{E} J_{\xi} \left(1, \Gamma_{1}\right)\label{eqn:bellman_sum_power}
\end{eqnarray}
$c_p$ in (\ref{eqn:bellman_sum_power}) is the cost of placing a relay
at the state $(r,\gamma)$, and $c_{np}$ is the cost of not placing a relay. 

If the current state is $(r,\gamma)$ and the line has not ended yet, we can take either of the two actions. 
If we place a relay, a cost 
$(\xi+\gamma)$ is incurred; another cost $J_{\xi}(\mathbf{0})$ is also incurred since 
the decision process regenerates at the point. If we do not 
place a relay, the line will end with probability $\theta$ in the next step, in which case a cost $\mathbb{E} (\Gamma_{r+1})$ will be 
incurred. If the line does not end in the next step, the next state will be $(r+1,\gamma')$ where $\gamma' \sim G_{r+1}$ 
and a mean cost of 
$\mathbb{E} J_{\xi} (r+1, \Gamma_{r+1})=\sum_{\gamma}g(r+1,\gamma)J_{\xi} (r+1,\gamma)$ will be incurred.
Note that it is never optimal to place a relay at state $\mathbf{0}$. If it were so, then we would have placed infinitely many relays 
at the sink, leading to infinite relay cost. But if we place one relay at each step until the line ends, the expected cost will 
be less than $(\frac{1}{\theta×}+1)(\xi+\mathbb{E}(\Gamma_{1}))<\infty$. Hence, the optimal action at state $\mathbf{0}$ 
would be to move to the next step without placing the relay. In the next step the line ends with probability $\theta$, in which case 
a cost $\mathbb{E} (\Gamma_{1})$ is incurred. If the line does not end in the next step, the next state will be $(1,\gamma)$ where 
$\gamma \sim G_{1}$.

{\em Justification for restricting to the class of stationary, deterministic policies:} From (\ref{eqn:bellman_sum_power}), 
we see that for each state, there is one action from the set of actions that achieves the minimum in the 
Bellman equation. Hence, by Proposition~$3.1.3$ of \cite{bertsekas07dynamic-programming-optimal-control-2}, 
we have a stationary deterministic optimal policy. Hence, we can 
focus on the class of stationary, deterministic policies.

\subsubsection{Value Iteration}\label{subsubsec:value_iteration_sum_power_adjacent} The 
value iteration for (\ref{eqn:unconstrained_total_power_problem}) is given by: 
\vspace{-2mm}
\begin{eqnarray}
J_{\xi}^{(k+1)}(\mathbf{0})&=&\theta \mathbb{E} (\Gamma_{1}) + (1-\theta) \mathbb{E} J_{\xi}^{(k)} \left(1, \Gamma_{1}\right)\nonumber\\
 J_{\xi}^{(k+1)}(r,\gamma)&=&\min \bigg\{ \xi+ \gamma + J_{\xi}^{(k)}(\mathbf{0}), \theta \mathbb{E} (\Gamma_{r+1}) \nonumber\\
&+& (1-\theta) \mathbb{E} J_{\xi}^{(k)}\left(r+1, \Gamma_{r+1}\right)  \bigg\} \label{eqn:value_iteration_sum_power}
\end{eqnarray}
where $J_{\xi}^{(0)}(r,\gamma)=0$ for all $r, \gamma$ and $J_{\xi}^{(0)}(\mathbf{0})=0$.

\begin{lem}\label{lemma:convergence_value_iteration_sum_power}
 The value iteration (\ref{eqn:value_iteration_sum_power}) provides a nondecreasing sequence of iterates that converges to the  
optimal value function, i.e., 
$J_{\xi}^{(k)}(r,\gamma) \uparrow J_{\xi}(r,\gamma)$ for all $r,\gamma$, and 
$J_{\xi}^{(k)}(\mathbf{0}) \uparrow J_{\xi}(\mathbf{0})$ as $k \uparrow \infty$.
\end{lem}
\begin{proof}
 See Appendix \ref{appendix:only_adjacent_nodes}.
\end{proof}

\subsubsection{Policy Structure}\label{subsubsec:policy_structure_sum_power_adjacent}

\begin{lem}\label{lem:properties_value_function_sum_power}
 $J_{\xi}(r, \gamma)$ is concave, increasing in $\gamma$ and $\xi$ and also increasing in $r$. $J_{\xi}(\mathbf{0})$ 
is concave, increasing in $\xi$.
\end{lem}
\begin{proof}
  See Appendix \ref{appendix:only_adjacent_nodes}.
\end{proof}

\begin{thm}{\em Policy Structure:}\label{theorem:policy_structure_sum_power}
 The optimal policy for Problem (\ref{eqn:unconstrained_total_power_problem}) is a threshold policy with a threshold $\gamma_{\textit{th}}(r)$ 
increasing in $r$ 
such that at a state $(r, \gamma)$ it is optimal to place a relay if and only if $\gamma \leq \gamma_{\textit{th}}(r)$. 
This corresponds to the condition $c_p \leq c_{np}$.
\end{thm}
\begin{proof}
   See Appendix \ref{appendix:only_adjacent_nodes}.
\end{proof}
\vspace{1mm}
\begin{rem}
 If $\gamma=\gamma_{\textit{th}}(r)$, either action is optimal.
\end{rem}
\vspace{1mm}
{\em Discussion of the Policy Structure:} We do not place at an $r$ if the required power is too high, 
as one might expect to get a better channel if 
one takes another step. For each $r$, there is a threshold on $\gamma$ below which we place. This threshold increases with $r$ 
since the distribution of $\Gamma_r$ is stochastically increasing with $r$.

Note that the optimal policy in Theorem~\ref{theorem:policy_structure_sum_power} can also be stated as follows: place a relay if and 
only if $r \leq r_{th}(\gamma)$ (i.e., $c_p \leq c_{np}$) where $r_{th}(\gamma)$ is some threshold on $r$, increasing in $\gamma$. 
Moreover, if there is a function $d(\cdot)$ such that $g(r,d(r))=1$ for all $r$, $d(r)$ is convex 
increasing in $r$, $d(r)\uparrow \infty$ and $d'(r)\uparrow \infty$ as $r \uparrow \infty$, then we have the same problem as 
\cite{mondal-etal12impromptu-deployment_NCC}, in which case it is optimal to place if and only if $r \geq r_{th}$.

\subsubsection{Computation of the Optimal Policy}\label{subsubsec:policy_computation_sum_power_adjacent} Let 
us write $V_{\xi}(r):=\mathbb{E} J_{\xi}\left(r, \Gamma_{r} \right)$, i.e., 
$V_{\xi}(r):=\sum_{\gamma} g(r,\gamma) J_{\xi}\left(r, \gamma \right)$ for all $r \in \{1,2,3,\cdots\}$, and  
$V_{\xi}(\mathbf{0}):=J_{\xi}(\mathbf{0})$. Also, for each stage $k \geq 0$ of the value iteration (\ref{eqn:value_iteration_sum_power}), 
define $V_{\xi}^{(k)}(r):=\mathbb{E} J_{\xi}^{(k)}\left(r, \Gamma_{r} \right)$ and 
$V_{\xi}^{(k)}(\mathbf{0}):=J_{\xi}^{(k)}(\mathbf{0})$. 


Observe that from the value iteration (\ref{eqn:value_iteration_sum_power}), we obtain:
\begin{eqnarray}
 V_{\xi}^{(k+1)}(r)&=&\sum_{\gamma} g(r, \gamma) \min \bigg\{ \xi+ \gamma + V_{\xi}^{(k)}(\mathbf{0}),  \nonumber\\
& &\theta \mathbb{E} (\Gamma_{r+1}) + (1-\theta) V_{\xi}^{(k)}(r+1)  \bigg\} \nonumber\\
V_{\xi}^{(k+1)}(\mathbf{0})&=&\theta \mathbb{E} (\Gamma_{1}) + (1-\theta) V_{\xi}^{(k)} (1)
\label{eqn:new_iteration_sum_power}
\end{eqnarray}
Since $J_{\xi}^{(k)}(r,\gamma) \uparrow J_{\xi}(r,\gamma)$ for each $r$, $\gamma$ and 
$J_{\xi}^{(k)}(\mathbf{0}) \uparrow J_{\xi}(\mathbf{0})$ as $k \uparrow \infty$, we can argue that 
$V_{\xi}^{(k)}(r)\uparrow \mathbb{E}J_{\xi}(r, \Gamma (r))$ for all 
$r \in \{1,2,3,\cdots\}$ (by Monotone Convergence Theorem) and  
$V_{\xi}^{(k)}(\mathbf{0})\uparrow J_{\xi}(\mathbf{0})$. Thus, 
$V_{\xi}^{(k)}(r)\uparrow V_{\xi}(r)$ and $V_{\xi}^{(k)}(\mathbf{0}) \uparrow V_{\xi}(\mathbf{0})$. 
 Hence, by the function iteration (\ref{eqn:new_iteration_sum_power}), we obtain $V_{\xi}(\mathbf{0})$ and 
$V_{\xi}(r)$ for all $r \geq 1$. Then, from (\ref{eqn:bellman_sum_power}), 
we can compute $\gamma_{\textit{th}}(r)$. Thus, for this iteration, we need not keep track of the cost-to-go values 
$J_{\xi}^{(k)}(r, \gamma)$ at each stage $k$; we simply need to keep track of $V_{\xi}^{(k)}(r)$.

\begin{figure*}[t!]
\begin{eqnarray}
 J_{\xi}(r,\gamma,\gamma_{\textit{max}})
&=&\min \bigg\{ \underbrace{\xi+ \theta \mathbb{E} \max\{\gamma, \gamma_{\textit{max}}, \Gamma_{1}\}
+ (1-\theta) \mathbb{E}J_{\xi}(1,\Gamma_{1},
\max\{\gamma, \gamma_{\textit{max}}\})}_{c_p},\nonumber\\ 
&& \underbrace{\theta \mathbb{E} \max \{\gamma_{\textit{max}}, \Gamma_{r+1}\}+(1-\theta)
  \mathbb{E}J_{\xi}(r+1, \Gamma_{r+1}, \gamma_{\textit{max}})}_{c_{np}} \bigg\}\label{eqn:bellman_max_power}
\end{eqnarray}
\begin{eqnarray}
 J_{\xi}^{(k+1)}(r,\gamma,\gamma_{\textit{max}})=
\min & \bigg\{ & \xi+ \theta \mathbb{E} \max\{\gamma, \gamma_{\textit{max}}, \Gamma_{1}\}
+ (1-\theta) \mathbb{E}J_{\xi}^{(k)}(1,\Gamma_{1}, \max\{\gamma, \gamma_{\textit{max}}\}), \nonumber\\
& & \theta \mathbb{E} \max \{\gamma_{\textit{max}}, \Gamma_{r+1}\}+(1-\theta) 
  \mathbb{E}J_{\xi}^{(k)}(r+1, \Gamma_{r+1}, \gamma_{\textit{max}}) \bigg\}\label{eqn:value_iteration_max_power}
\end{eqnarray}
\hrule
\end{figure*}

\subsubsection{A Numerical Example}\label{subsubsec:numerical_sum_power_adjacent} 
We take $\delta =2$ meters and $\theta = 0.025$, (i.e., $\overline{L}
= 40$~steps, or $80$~meters), and
$\mathcal{S}=\{-25,-20,-15,-10,-5,0,3\}$ in dBm. Using a standard
model, with transmit power $P_T$ (mW), the received power (in mW) at a
distance $r$ from the transmitter is given by $P_T \alpha
(\frac{r}{r_0×})^{-\eta} H 10^{-\frac{\nu}{10×}}$ where $\alpha$ is a
constant and $r_0$ is a reference distance.  $H$ models Rayleigh
fading, and is exponentially distributed with mean $1$.  $\nu$ is
assumed to be distributed as $\mathcal{N}(0,\sigma^2)$ with $\sigma=8$
dB; i.e., we have log-normal shadowing. The shadowing is assumed to be
independent from link to link. For a commercial implementation of the
PHY/MAC of IEEE~802.15.4 (a popular wireless sensor networking
standard), $-88$~dBm received power corresponds to a $2\%$ packet loss
probability for $140$ byte packets. Taking $-88$~dBm to be the minimum
acceptable received power, we set the target received power (averaged
over fading) to be $\psi=10^{-7.5}$~mW (i.e., $-75$~dBm), which (under
Rayleigh fading) yields an ``outage'' probability of $5$\%.  Hence,
the transmit power required to establish the link is:
\begin{equation}
 P_{req}=\frac{\psi }{10^{-\frac{\nu}{10×}} \alpha×}\left({\frac{r}{r_0×}}\right)^{\eta} \label{eqn:reqd_transmit_power} 
\end{equation}

Now, since the set $\mathcal{S}$ is bounded, at large distance the required power to achieve $\psi$ will exceed $2$ mW 
(i.e., $3$ dBm) with high 
probability. To tackle this problem, we take the following 
approach. {\em We modify the problem formulation by 
requiring that a relay is placed when $r=10$ (in steps).}
For $1 \leq r \leq 10$ (in steps), we obtain $\Gamma_r$ by inverting the path-loss formula and using the next higher 
power level in $\mathcal{S}$. For example, if the required transmit power to establish the link,
obtained from (\ref{eqn:reqd_transmit_power}), is in $(-5,0]$ dBm, 
then we say that the power requirement is $0$ dBm for that link. 
But, if the power requirement is more than $3$ dBm, then we say that the required power is $3$ dBm. 
It is easy to see that the distribution $G_{r}(\cdot)$ of $\Gamma_r$ obtained in this fashion is stochastically increasing
\footnote{The results asserted for the formulation in 
Section~\ref{subsec:sum_power_adjacent}, e.g., threshold structure of the optimal policy, 
hold for this variation as well. The only change will be that at $r=10$ 
steps the only feasible action is to place the relay. We can use the same iterations as 
(\ref{eqn:value_iteration_sum_power}) and (\ref{eqn:new_iteration_sum_power}) to compute the optimal policy in this case, 
except that we need to set $J_{\xi}^{(k)}(r,\gamma)$ or $V_{\xi}^{(k)}(r)$ to $\infty$ at each iteration if $r$ is $11$ steps 
or more. This will force the policy to place the relay within every $10$ steps.} in $r$. 
For the numerical work, we assume that $\eta=2.5$, $\frac{\alpha}{r_0^{-\eta}×}=10^{-3}/m^{-2.5}$ ($-30$ dB). 
With these parameters, the probablity of the transmit power required to achieve a target received power $\psi=-75$ dBm 
(averaged over fading) for a link of length $10$ steps ($20$ meters) exceeding $3$ dBm (which is the maximum transmit power) 
will be less than $2.65\%$. This probability 
will be even less for links having length smaller than $20$ meters. 
Note that these comment imply that there is a positive probability of deployment failure; 
we will quantify this later in Section~\ref{subsubsec:deployment_failure_sum_power_adjacent}.

Figure~\ref{fig:sum_power_gamma_th_vs_r} shows  
the variation of $\gamma_{\textit{th}}(r)$ with $r$ and $\xi$. Here the unit of $\xi$ is mW/relay, so that 
the unit of $J_{\xi}(\mathbf{0})$ is also mW. 
We see from Figure~\ref{fig:sum_power_gamma_th_vs_r} that, for $\xi = 0.001$, if the previous relay is $8$~meters behind, and the 
required power is $-20$ dBm or below, then a relay should be placed at this point because $\gamma_{th}(r=8 \, meters)=-20 dBm$. 
Also, note that $-25$ dBm 
is the smallest possible transmit power level, and for $\xi=0.001$ we will never place a relay at $r=2$ meters since there 
$\gamma_{th}(r)$ is $0$ mW ($-\infty$ dBm).
The variation of $\gamma_{\textit{th}}(r)$ with $r$ has already been 
established in Theorem~\ref{theorem:policy_structure_sum_power}. 
But Figure~\ref{fig:sum_power_gamma_th_vs_r} also 
shows that $\gamma_{\textit{th}}(r)$ is decreasing in $\xi$ for each $r$. This is intuitive because $\xi$ is the price of a relay, 
and all it says is that as the price of a relay increases, we will place the relays less frequently.

The variation of the mean number of relays $\mathbb{E}(N)$ and various cost components with $\xi$ is shown in 
Table~\ref{table:sum_power_adjacent}. It shows that as the cost of a relay $\xi$ increases, the mean number of relays decreases, and 
the power cost and $J_{\xi}(\mathbf{0})$ increase.


\begin{table}[t!]
\centering
\begin{tabular}{|c |c |c |c|}
\hline
                 & $\xi=0.001$ & $\xi=0.01$ & $\xi=0.1$ \\ \hline
$\mathbb{E}(N)$  & 15.8754  &   10.3069  &  5.3225 \\ \hline
Relay cost $\xi \mathbb{E}(N)$ & 0.01588 &  0.10307  &   0.53225  \\ \hline
Power Cost  &  0.07513 & 0.08277  &  0.19291 \\ \hline
$J_{\xi}(\mathbf{0})$  & 0.09101  &  0.18584 & 0.72516\\ \hline
\end{tabular}
\caption{Relaying via the last placed relay: break-up of the optimal cost for 
the example in Section~\ref{subsec:sum_power_adjacent}, for three values of the relay cost $\xi$.}
\vspace{-0.8cm}
\label{table:sum_power_adjacent}
\end{table}

\begin{figure}[!t]
\centering
\includegraphics[scale=0.27]{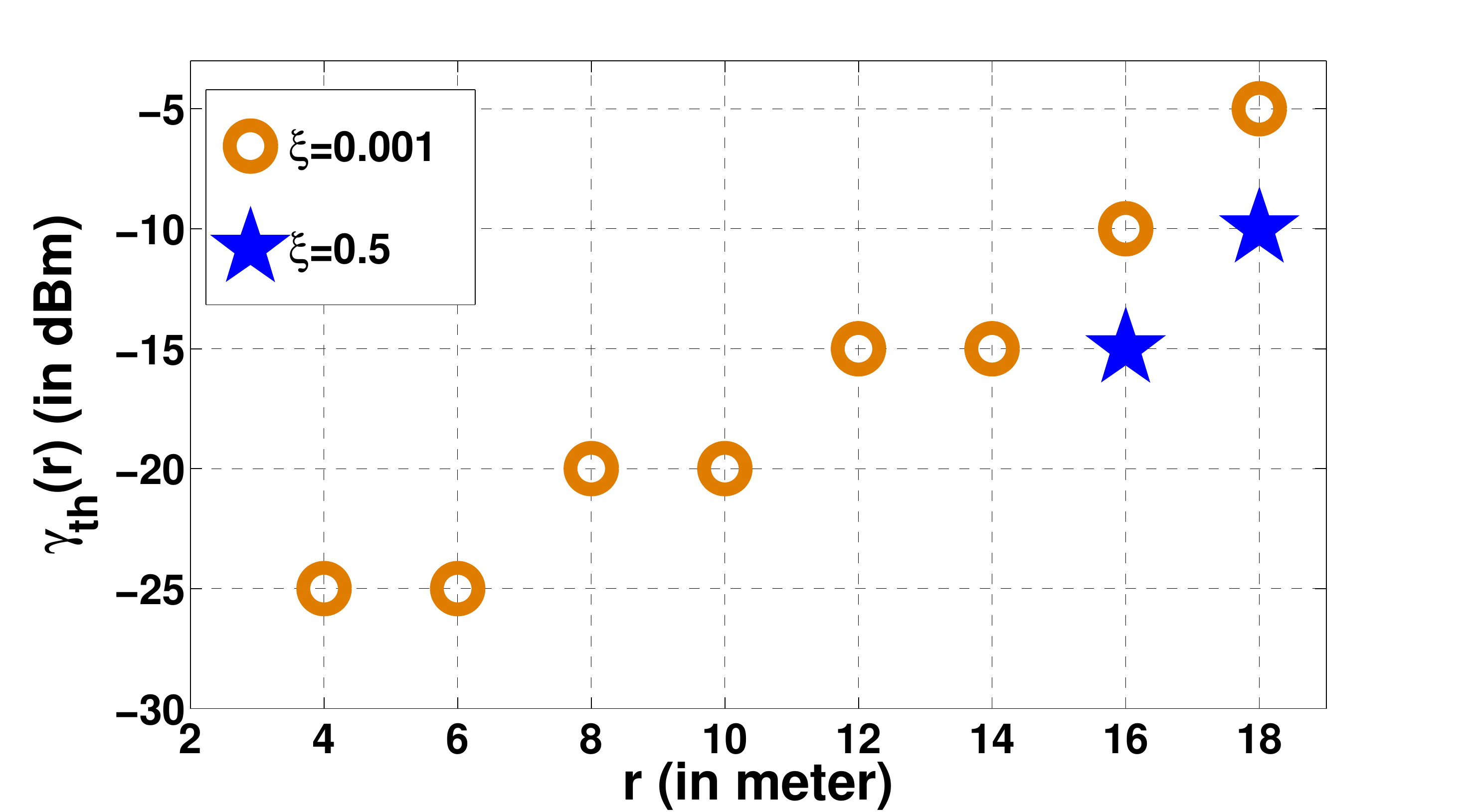}
\vspace{-2mm}
\caption{Relaying only via last placed relay: $\gamma_{th}(r)$ vs.\ $r$, for various relay costs, $\xi$,  
for the numerical example described in Section~\ref{subsubsec:numerical_sum_power_adjacent}.}
\label{fig:sum_power_gamma_th_vs_r}
\vspace{-4mm}
\end{figure}

\subsubsection{Deployment Failure}\label{subsubsec:deployment_failure_sum_power_adjacent}

Since, in practice, there is a maximum power at which a transmitter can 
transmit (e.g., 3~dBm), there is a possibility that a deployment can 
fail. It is interesting to compute the probability of such failure in 
the algorithms that we have derived. Here we provide a simulation 
estimate of the deployment failure probabilities for the sum power objective, under the threshold policies 
obtained numerically in Section~\ref{subsubsec:numerical_sum_power_adjacent}. In our numerical 
example in Section~\ref{subsubsec:numerical_sum_power_adjacent}, deployment failure can occur in the 
following two cases: (i) at $r=10$ steps the required power exceeds 3 dBm; the probability of this event is 
$2.65\%$ , and (ii) the source at the end of the line 
requires more than 3~dBm power. By simulating 200000 deployments, we observe that the 
deployment failure probability for $\xi=0.001,0.01,0.1$ and $1$ are $0.025\%,0.057\%,0.555\%$ and $3.8\%$ respectively. 

Evidently, there is a trade-off between the target link performance and 
the probability of link failure. In addition, by placing relays more 
frequently, we reduce the chance of being caught in a situation where 
the deployment operative has walked too far without placing a relay and is 
unable to get a workable link to the previous node. In future work, we 
propose to include deployment failure probability as a constraint in the 
optimization formulation. Another way to reduce deployment failure is to 
permit {\em back-tracking} by the deployment operative, which, of course, will 
require the placement algorithm to keep more measurement history; we 
propose to permit this in our future work as well.

\vspace{-1mm}
\subsection{Max-Power Objective}\label{subsec:max_power_adjacent}
\vspace{-2mm}
\subsubsection{Problem Formulation}\label{subsubsec:formulation_max_power_adjacent}
We aim to address the following problem:
\begin{eqnarray}
 \min_{\pi \in \Pi} \mathbb{E}_{\pi} \bigg (\max_{i \in \{1,2,\cdots,N+1\}} \Gamma^{(i,i-1)}+\xi N \bigg) 
\label{eqn:unconstrained_max_power_problem} 
\end{eqnarray}
We formulate (\ref{eqn:unconstrained_max_power_problem}) as an MDP.  
The state of the system is $(r,\gamma,\gamma_{\textit{max}})$ where $r$ and $\gamma$ are the same as before, and $\gamma_{\textit{max}}$ 
is the maximum power used in all the previously established links. The action space is 
$\{\textit{place}, \, \textit{do not place}\}$ as before. The cost structure is such that 
the power cost is incurred only after the source node is placed.

\subsubsection{Bellman Equation}\label{subsubsec:bellman_max_power_adjacent}
The problem is again an infinite horizon total cost problem with countable state space, finite action space and 
nonnegative single-stage cost. Hence, 
by the same arguments as used in problem (\ref{eqn:unconstrained_total_power_problem}), the optimal value 
function $J_{\xi}(\cdot)$ satisfies the Bellman equation (\ref{eqn:bellman_max_power}). At state $\mathbf{0}$, it is not 
optimal to place a relay. Hence, $J_{\xi}(\mathbf{0})=\theta \mathbb{E}(\Gamma_1)+(1-\theta)\mathbb{E}J_{\xi}(1,\Gamma_1,0)$.

At state $(r,\gamma,\gamma_{\textit{max}})$, if we place a relay, we incur a cost $\xi$ and in the next step the line ends 
with probability $\theta$ in which case a power cost of $\mathbb{E} \max\{\gamma, \gamma_{\textit{max}}, \Gamma_{1}\}$ is 
incurred. If the line does not end in the next step, the next state becomes $(1,\gamma', \max\{\gamma, \gamma_{\textit{max}}\})$ 
where $\gamma' \sim G_{1}$, and a cost of 
$\mathbb{E}J_{\xi}(1,\Gamma_{1}, \max\{\gamma, \gamma_{\textit{max}}\})$ is incurred. On the other hand, if we do not place a 
relay at state $(r,\gamma,\gamma_{\textit{max}})$, the line ends in the next step 
with probability $\theta$ in which case a power cost of $\mathbb{E} \max\{\gamma_{\textit{max}}, \Gamma_{r+1}\}$ is 
incurred. If the line does not end in the next step, the next state will be $(r+1, \gamma', \gamma_{\textit{max}})$ where 
$\gamma' \sim G_{r+1}$.

\subsubsection{Value Iteration}\label{subsubsec:value_iteration_max_power_adjacent}
The value iteration for this MDP is given by (\ref{eqn:value_iteration_max_power})
with $J_{\xi}^{(0)}(r,\gamma,\gamma_{\textit{max}})=0$ for all $r$, $\gamma$, $\gamma_{\textit{max}}$.

\begin{lem}\label{lemma:convergence_value_iteration_max_power}
 The iterates of the value iteration (\ref{eqn:value_iteration_max_power}) converge 
to the optimal value function, i.e., 
$J_{\xi}^{(k)}(r,\gamma,\gamma_{\textit{max}}) \uparrow J_{\xi}(r,\gamma,\gamma_{\textit{max}})$ for all 
$(r,\gamma,\gamma_{\textit{max}})$, as $k \uparrow \infty$.
\end{lem}
\begin{proof}
 See Appendix \ref{appendix:only_adjacent_nodes}.
\end{proof}

\subsubsection{Policy Structure}\label{subsubsec:policy_structure_max_power_adjacent}

\begin{lem}\label{lem:properties_value_function_max_power}
 $J_{\xi}(r, \gamma, \gamma_{\textit{max}})$ is concave, increasing in $\xi$ and increasing in $r$, $\gamma$, $\gamma_{\textit{max}}$. 
\end{lem}
\begin{proof}
 See Appendix \ref{appendix:only_adjacent_nodes}.
\end{proof}

\begin{thm}{\em Policy Structure:}\label{theorem:policy_structure_max_power}
The conditions for optimal relay placement are:
\begin{enumerate}[label=(\roman{*})]
\item {\em If $\gamma \leq \gamma_{\textit{max}}$,} place the relay when $r \geq r_{\textit{th}}(\gamma_{\textit{max}})$ where 
$r_{\textit{th}}(\gamma_{\textit{max}})$ is a threshold value.
 \item {\em If $\gamma>\gamma_{\textit{max}}$,} place the relay when 
$\gamma \leq \gamma_{\textit{th}}(r, \gamma_{\textit{max}})$ where 
$\gamma_{\textit{th}}(r, \gamma_{\textit{max}})$ is 
a threshold value increasing in $r$ and $\gamma_{\textit{max}}$.
\end{enumerate}
 \end{thm}
\begin{proof}
 See Appendix \ref{appendix:only_adjacent_nodes}.
\end{proof}

{\em Discussion of the policy structure:}
When $\gamma \leq \gamma_{\textit{max}}$, we can postpone placement until the point beyond which the chance of 
getting a worse value of power becomes significant. For $\gamma > \gamma_{\textit{max}}$, waiting to place the relay may 
result in a better channel; there is a threshold $\gamma_{\textit{th}}(r, \gamma_{\textit{max}})$ 
such that $\gamma_{\textit{th}}(r, \gamma_{\textit{max}})$ may cross $\gamma_{\textit{max}}$ for large enough $r$. 
If $\gamma$ is between these two values then we place.

\subsubsection{Computation of the Optimal Policy}\label{subsubsec:policy_computation_max_power_adjacent}
 Let us define $V_{\xi}(r,\gamma_{\textit{max}}):=\mathbb{E} J_{\xi}(r,\Gamma_{r},\gamma_{\textit{max}})$. 
We can again argue that the following function iteration (similar to that used in Section~\ref{subsec:sum_power_adjacent})) 
will yield $V_{\xi}(r,\gamma_{\textit{max}})$ for all 
$r$, $\gamma_{\textit{max}}$, from which we can compute 
$r_{\textit{th}}(\gamma_{\textit{max}})$ and $\gamma_{\textit{th}}(r, \gamma_{\textit{max}})$:

\footnotesize
\begin{eqnarray}
 V_{\xi}^{(k+1)}(r,\gamma_{\textit{max}})&=& \sum_{\gamma} g(r,\gamma) \min \bigg\{ \xi+ \theta \mathbb{E} \max\{\gamma, \gamma_{\textit{max}}, \Gamma_{1}\} \nonumber\\
&+& (1-\theta) V_{\xi}^{(k)}(1, \max\{\gamma, \gamma_{\textit{max}}\}), \nonumber\\
&&\theta \mathbb{E} \max \{\gamma_{\textit{max}}, \Gamma_{r+1}\}\nonumber\\
&+& (1-\theta) V_{\xi}^{(k)}(r+1, \gamma_{\textit{max}}) \bigg\} \label{eqn:function_iteration_max_power}
\end{eqnarray}
\normalsize
with $V_{\xi}^{(0)}(r,\gamma_{\textit{max}})=0$ for all $r$, $\gamma_{\textit{max}}$.

\begin{table}[t!]
\centering
\begin{tabular}{|c |c |c |c|}
\hline
 & $\xi=0.001$ & $\xi=0.01$ & $\xi=0.1$ \\ \hline
$\mathbb{E}(N)$  & 18.1178 &  8.6875   &  4.6615 \\ \hline
Relay Cost  & 0.01812 &  0.08688  &   0.46615  \\ \hline
Power Cost  & 0.01524  &  0.04436  & 0.15079 \\ \hline
$J_{\xi}(\mathbf{0})$  & 0.03336 & 0.13124  &  0.61693\\ \hline
\end{tabular}
\caption{Relaying via the last placed relay: break-up of the optimal cost for 
the example in Section~\ref{subsec:max_power_adjacent}, for three values of the relay cost $\xi$.}
\vspace{-0.5cm}
\label{table:max_power_adjacent}
\end{table}

\subsubsection{A Numerical Example}\label{subsubsec:numerical_max_power_adjacent} Figure~\ref{fig:max_power_r_th_vs_gamma_max} 
shows the variation of $r_{\textit{th}}(\gamma_{\textit{max}})$ 
with $\gamma_{\textit{max}}$ and $\xi$. Here we consider the same setting as 
Section~\ref{subsubsec:numerical_sum_power_adjacent}. 
The plot shows that $r_{\textit{th}}(\gamma_{\textit{max}})$ increases 
with $\gamma_{\textit{max}}$. To get an insight into the reason, let us 
consider the situation $\gamma < \gamma_{\textit{max}}$. If $r$ is small, then it is more likely 
that in the next step also the power required to establish a link to the last node will be below $\gamma_{\textit{max}}$, 
and hence, we don't need to place a relay. But 
if $r$ is large, then it is more likely that the required power will cross $\gamma_{\textit{max}}$ in the next step, and hence we will have a threshold 
$r_{\textit{th}}(\gamma_{\textit{max}})$ beyond which we have to place the relay. As $\gamma_{\textit{max}}$ increases, the 
probability that the power required to establish a link to the last node exceeding $\gamma_{\textit{max}}$ decreases for 
each $r$, thereby 
increasing $r_{\textit{th}}(\gamma_{\textit{max}})$. Also, $r_{\textit{th}}(\gamma_{\textit{max}})$ increases with $\xi$ because if 
the price of a relay increases, we will place relays less frequently.

Figure~\ref{fig:max_power_gamma_th_vs_r_max_power} shows the variation of $\gamma_{th}(r, \gamma_{max})$ with $r$ for 
$\gamma_{max}=-20$ dBm and two 
different values of $\xi$. 
For very small $\xi$ (e.g., $\xi=0.00001$), $\gamma_{th}(r, \gamma_{max})$ can be more than $\gamma_{max}$ for all $r$. For moderate 
values of $\xi$ (e.g., $\xi=0.01$ in Figure~\ref{fig:max_power_r_th_vs_gamma_max}), $\gamma_{th}(r, \gamma_{max})$ vs. $r$ curve 
crosses $\gamma_{max}$ at $r=r_{th}(\gamma_{max})$. Also, we have seen numerically that for $\xi$ very large, 
$\gamma_{th}(r, \gamma_{max})$ is $0$ mW for $r \leq 9$ steps; for large $\xi$ we will place a relay only when $r=10$ steps 
($20$~meters). 

The variation of the mean number of relays $\mathbb{E}(N)$ and different cost components with $\xi$ is shown in 
Table~\ref{table:max_power_adjacent}. It shows that as the cost of a relay $\xi$ increases, the mean number of relays decrease, and 
the power cost and $J_{\xi}(\mathbf{0})$ increase. Note that, for any given deployment and any given relay cost $\xi$, the sum power 
is always greater than the max power in the network. Hence, for a given $\xi$, the mean power cost and $J_{\xi}(\mathbf{0})$ for the sum power objective 
will be greater than the corresponding values for the max-power objective, as seen in Table~\ref{table:sum_power_adjacent} and 
Table~\ref{table:max_power_adjacent}.

The estimated probability of deployment failure obtained from 200000 simulations of the deployment for 
$\xi=0.001,0.01,0.1$ and $1$ are $0.01\%,0.145\%,0.73\%$ and $5.39\%$ respectively.

\begin{figure}[!t]
\centering
\includegraphics[scale=0.27]{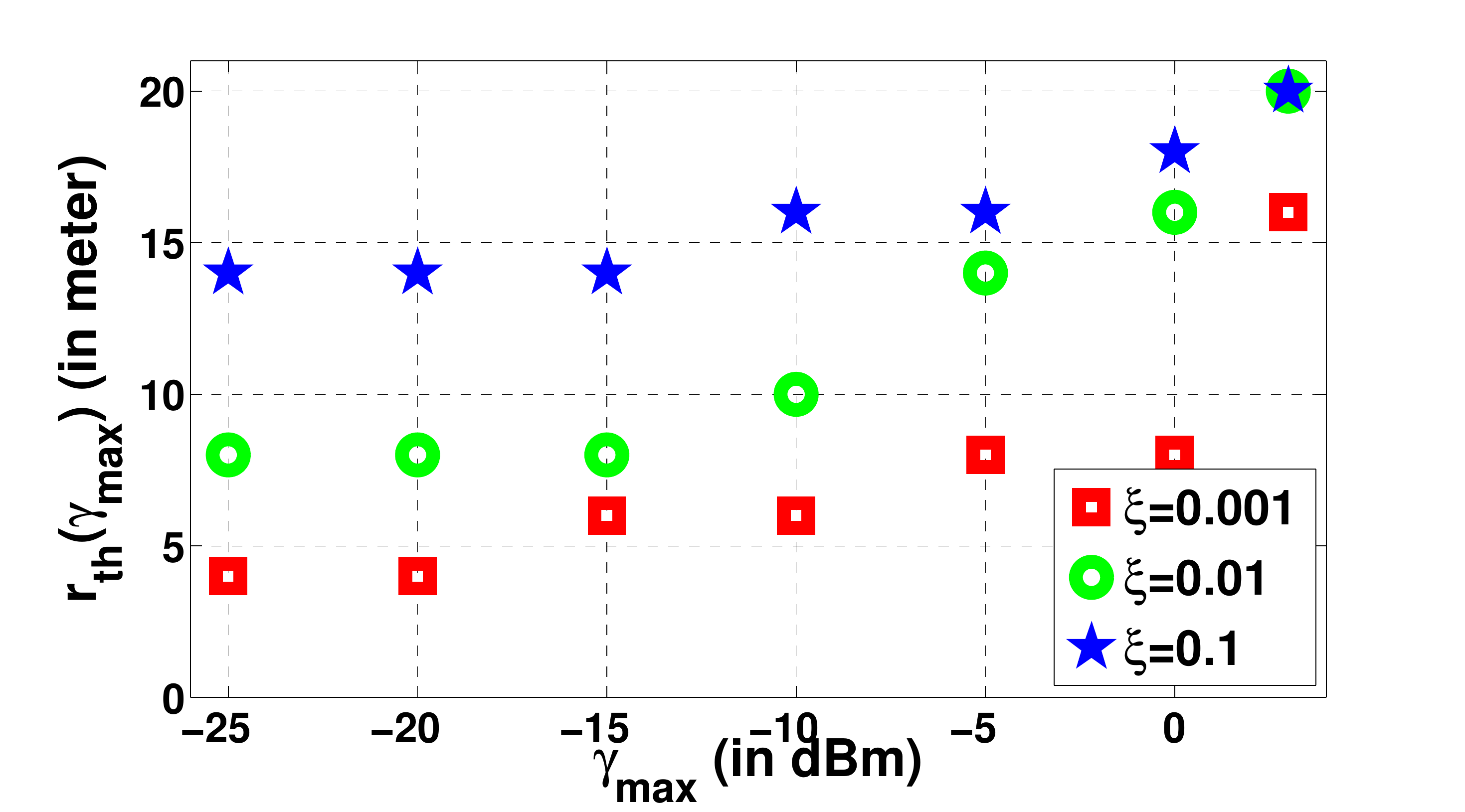}
\vspace{-2mm}
\caption{Relaying only via last placed relay: $r_{\textit{th}}(\gamma_{\textit{max}})$ 
vs. $\gamma_{\textit{max}}$, for various relay costs, $\xi$,  
for the numerical example described in Section~\ref{subsubsec:numerical_max_power_adjacent}.}
\label{fig:max_power_r_th_vs_gamma_max}
\vspace{-1mm}
\end{figure}

\begin{figure}[!t]
\centering
\includegraphics[scale=0.27]{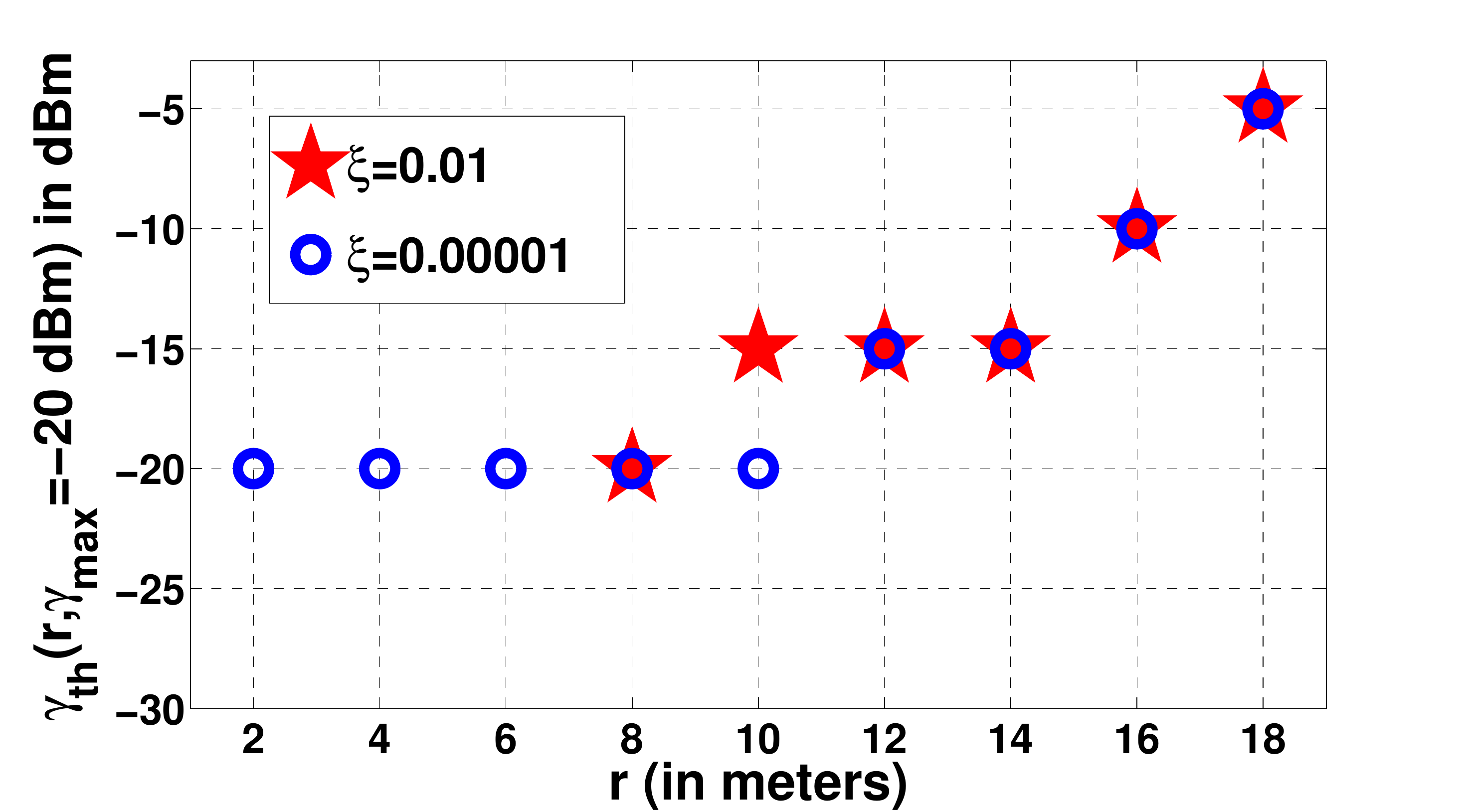}
\vspace{-2mm}
\caption{Relaying only via last placed relay: $\gamma_{th}(r,\gamma_{max})$ vs. $r$, for $\gamma_{max}=-20$ dBm and 
various relay costs, $\xi$,  
for the numerical example described in Section~\ref{subsubsec:numerical_max_power_adjacent}.}
\label{fig:max_power_gamma_th_vs_r_max_power}
\vspace{-1mm}
\end{figure}

\section{Relaying via Any Previous Node}
\label{sec:not_only_adjacent_nodes}


In Section \ref{sec:only_adjacent_nodes}, we considered the case where, after the deployment is over, 
only the links between adjacent nodes are permitted, i.e., only the links represented by the solid lines in 
Figure~\ref{fig:line-network-general} can be used. However, as discussed in Section~\ref{sec:system_model_and_notation}, 
while formulating the problem we need to take into account the fact that some relays might be skipped after deployment, 
i.e., some of the links represented by the dotted lines in Figure~\ref{fig:line-network-general} can be used. 
This section is dedicated to such formulation and exploration of the 
structural properties of the relay placement policies for different objectives.

\subsection{Sum-Power Objective}

\subsubsection{Problem Definition}\label{subsubsection:problem_sum_power_shortest_path}

\begin{figure*}[!t]
\footnotesize
\begin{eqnarray}
& & J_{\xi}\bigg(\{y_k\}_{k=1}^n; \{P^{(k)}\}_{k=1}^n; \{\gamma^{(k)}\}_{k=1}^n \bigg)
= \min \bigg\{ \xi+ \theta \mathbb{E} \min \bigg\{\min_{k \in \{1,2,\cdots,n-1\}} (\Gamma_{y_k+1}+P^{(k)}), 
\Gamma_{1}+ \min_{k \in \{1,2,\cdots,n\}} (\gamma^{(k)}+P^{(k)}) \bigg\}\nonumber\\
&+& (1-\theta) \mathbb{E} J_{\xi} \bigg(1,y_1+1,\cdots,y_{n-1}+1;\min_{k \in \{1,\cdots,n\}} (\gamma^{(k)}+P^{(k)}), P^{(1)},\cdots,P^{(n-1)}
; \Gamma_{1}, \Gamma_{y_1+1},\cdots,\Gamma_{y_{n-1}+1}\bigg), \nonumber\\
&& \theta \mathbb{E} \min_{k \in \{1,\cdots,n\}} \bigg(\Gamma_{y_k+1}+P^{(k)}\bigg)
+(1-\theta) \mathbb{E} J_{\xi} \bigg( \{y_k+1\}_{k=1}^n; \{P^{(k)}\}_{k=1}^n; \{\Gamma_{y_k+1}\}_{k=1}^n \bigg) \bigg\}
\label{eqn:bellman_sum_power_shortest_path}
\end{eqnarray}
\normalsize
\end{figure*}

\begin{figure*}[!t]
\vspace{-6mm}
\footnotesize
\begin{eqnarray}
& &J_{\xi}\bigg(\{y_k\}_{k=1}^m; \{P^{(k)}\}_{k=1}^m; \{\gamma^{(k)}\}_{k=1}^m \bigg)
= \min \bigg\{ \xi+ \theta \mathbb{E} \min \bigg\{\min_{k \in \{1,2,\cdots,m\}} (\Gamma_{y_k+1}+P^{(k)}), 
\Gamma_{1}+ \min_{k \in \{1,2,\cdots,m\}} (\gamma^{(k)}+P^{(k)}) \bigg\}\nonumber\\
&+& (1-\theta) \mathbb{E} J_{\xi} \bigg(1,y_1+1,\cdots,y_{m}+1;\min_{k \in \{1,\cdots,m\}} (\gamma^{(k)}+P^{(k)}), P^{(1)},\cdots,P^{(m)}
; \Gamma_{1}, \Gamma_{y_1+1},\cdots,\Gamma_{y_{m}+1}\bigg), \nonumber\\
&& \theta \mathbb{E} \min_{k \in \{1,\cdots,m\}} \bigg(\Gamma_{y_k+1}+P^{(k)}\bigg)
+(1-\theta) \mathbb{E} J_{\xi} \bigg( \{y_k+1\}_{k=1}^m; \{P^{(k)}\}_{k=1}^m; \{\Gamma_{y_k+1}\}_{k=1}^m \bigg) \bigg\}
\label{eqn:bellman_sum_power_shortest_path_less_than_memory}
\end{eqnarray}
\normalsize
\hrule
\end{figure*}

Given a deployment of $N$ relays, indexed $1, 2, \cdots, N,$ consider the directed acyclic graph on these 
relays along with the sink (Node $0$) and the source (Node $N+1$), whose links are all directed edges 
from each node to every node with smaller index. 
Hence, if $i$ and $j$ are two nodes with $i>j$, there is only one link $(i,j)$ between them.
Consider all directed acyclic paths 
from the source to sink, on this graph.
Let us denote by $\mathbf{p}$ any arbitrary directed acyclic path from the source to the sink, 
and  by $\mathbf{E}(\mathbf{p})$ the 
set of (directed) links of the path $\mathbf{p}$.
We also define   
$\mathcal{P}_n:=\{\mathbf{p}:(i,j) \in \mathbf{E}(\mathbf{p}) \implies i>j, |i-j| \leq n\}$ a 
subcollection of paths between the source and the sink on the directed acyclic graph, 
such that no path in $\mathcal{P}_n$ contains a link between two nodes whose indices differ by a number larger that $n$. 
We call $n$ the ``memory'' of the class of policies we are considering.


Here we consider the following problem:
\begin{eqnarray}
 \min_{\pi \in \Pi} \mathbb{E}_{\pi} \bigg( \min_{\mathbf{p} \in \mathcal{P}_n} \sum_{e \in \mathbf{E}(\mathbf{p})} \Gamma^{(e)} + \xi N \bigg) 
\label{eqn:unconstrained_sum_power_problem_shortest}
\end{eqnarray}
where $\Gamma^{(e)}$ is the power used on the link $e$. 
We call $\sum_{e \in \mathbf{E}(\mathbf{p})} \Gamma^{(e)}$ the ``length" of the path $\mathbf{p}$, and 
$\min_{\mathbf{p} \in \mathcal{P}_n} \sum_{e \in \mathbf{E}(\mathbf{p})} \Gamma^{(e)}$ the length of the 
``shortest path" from the source to the sink (over the relays deployed by policy $\pi$ in a given 
realization of the decision process).

\begin{figure*}[!t]
\footnotesize
\begin{eqnarray}
& & J_{\xi}\bigg(\{y_k\}_{k=1}^n; \{P^{(k)}\}_{k=1}^n; \{\gamma^{(k)}\}_{k=1}^n \bigg)
= \min \bigg\{ \xi+ \theta \mathbb{E} \min \bigg\{\min_{k \in \{1,2,\cdots,n-1\}} \max\{\Gamma_{y_k+1},P^{(k)}\}, 
\max\{\Gamma_{1}, \min_{k \in \{1,2,\cdots,n\}} \max\{\gamma^{(k)},P^{(k)}\}\} \bigg\}\nonumber\\
&+& (1-\theta) \mathbb{E} J_{\xi} \bigg(1,y_1+1,\cdots,y_{n-1}+1;\min_{k \in \{1,\cdots,n\}} \max\{\gamma^{(k)},P^{(k)}\}, P^{(1)},\cdots,P^{(n-1)}
; \Gamma_{1}, \Gamma_{y_1+1},\cdots,\Gamma_{y_{n-1}+1}\bigg), \nonumber\\
&& \theta \mathbb{E} \min_{k \in \{1,\cdots,n\}} \max\{\Gamma_{y_k+1},P^{(k)}\}
+(1-\theta) \mathbb{E} J_{\xi} \bigg( \{y_k+1\}_{k=1}^n; \{P^{(k)}\}_{k=1}^n; \{\Gamma_{y_k+1}\}_{k=1}^n \bigg) \bigg\}
\label{eqn:bellman_max_power_shortest_path}
\end{eqnarray}
\normalsize
\end{figure*}

\begin{figure*}[!t]
\footnotesize
\vspace{-5mm}
\begin{eqnarray}
& &J_{\xi}\bigg(\{y_k\}_{k=1}^m; \{P^{(k)}\}_{k=1}^m; \{\gamma^{(k)}\}_{k=1}^m \bigg)
= \min \bigg\{ \xi+ \theta \mathbb{E} \min \bigg\{\min_{k \in \{1,2,\cdots,m\}} \max\{\Gamma_{y_k+1},P^{(k)}\}, 
\max\{\Gamma_{1}, \min_{k \in \{1,2,\cdots,m\}} \max\{\gamma^{(k)},P^{(k)}\}\} \bigg\}\nonumber\\
&+& (1-\theta) \mathbb{E} J_{\xi} \bigg(1,y_1+1,\cdots,y_{m}+1;\min_{k \in \{1,\cdots,m\}} \max\{\gamma^{(k)},P^{(k)}\}, P^{(1)},\cdots,P^{(m)}
; \Gamma_{1}, \Gamma_{y_1+1},\cdots,\Gamma_{y_{m}+1}\bigg), \nonumber\\
&& \theta \mathbb{E} \min_{k \in \{1,\cdots,m\}} \max\{\Gamma_{y_k+1},P^{(k)}\}
+(1-\theta) \mathbb{E} J_{\xi} \bigg( \{y_k+1\}_{k=1}^m; \{P^{(k)}\}_{k=1}^m; \{\Gamma_{y_k+1}\}_{k=1}^m \bigg) \bigg\}
\label{eqn:bellman_max_power_shortest_path_less_than_memory}
\end{eqnarray}
\normalsize
\hrule
\end{figure*}

\subsubsection{MDP Formulation}\label{subsubsection:mdp_formulation_sum_power_shortest_path}
Consider the evolution of the network as the relays are deployed.
Suppose that at some point in the deployment process there are $m$
preceding nodes, {\em including the sink}; see
Figure~\ref{fig:measurement-based-relay-placement} where $m=4$.  The transmit
power required to establish a link from the current location to the $k$-th previous node is
denoted by $\gamma^{(k)}$, and the distance of the current location
from the $k$-th previous node is denoted by $y_k$. 
Let $P^{(k)}_n$ denote the
length of the ``shortest path" from the $k$-th previous node to the
sink. We define $P^{(m)}_n:=0$ if $m \leq n$, i.e., the length of the
shortest path from the sink to itself is $0$ (when $m \leq n$, the
$m$-th previous node is the sink).  For notational simplicity, we drop
the subscript and denote $P^{(k)}_n$ by $P^{(k)}$. 
The deployment operative decides whether to place a node at his
current position based on (i) the powers $\gamma^{(1)}, \gamma^{(2)},\cdots,\gamma^{(n)}$, (ii) the distances 
$y_1,y_2,\cdots,y_n$, and 
(iii) the length of the shortest paths $P^{(1)},P^{(2)},\cdots,P^{(n)}$. If $n=2$, at the
``current location'' shown in
Figure~\ref{fig:shortest-path-measurement}, the decision will
be based on the powers $\gamma^{(1)}$, $\gamma^{(2)}$, the distances
$y_1$, $y_2$, and the shortest paths $P^{(1)}$ and $P^{(2)}$ at nodes $3$ and $2$ respectively. 
However, in case $m<n$, we do not have measurements for 
$n$ previous nodes. Hence, let us define $l_m:=\min\{m,n\}$. At each step, the deployment operative 
knows the distance $\{y_k\}_{k=1}^{l_m}$, the power $\{\gamma^{(k)}\}_{k=1}^{l_m}$ and the lengths of the shortest 
paths $\{P^{(k)}\}_{k=1}^{l_m}$. He decides based on this information whether to place a 
relay at the current position or not. We formulate this
problem as an MDP with state $(\{y_k\}_{k=1}^{l_m};\{P^{(k)}\}_{k=1}^{l_m};\{\gamma^{(k)}\}_{k=1}^{l_m})$, and the action
space $\{\textit{place}, \, \textit{do not place}\}$. The state at the sink is denoted by 
$\mathbf{0}$. Since the set $\mathcal{S}$ of transmit power levels is countable, $\{P^{(k)}\}_{k=1}^{l_m}$ 
also take values from a countable set. Hence, the state space is countable in our problem. 

\begin{figure}[!t]
\centering
\includegraphics[scale=0.30]{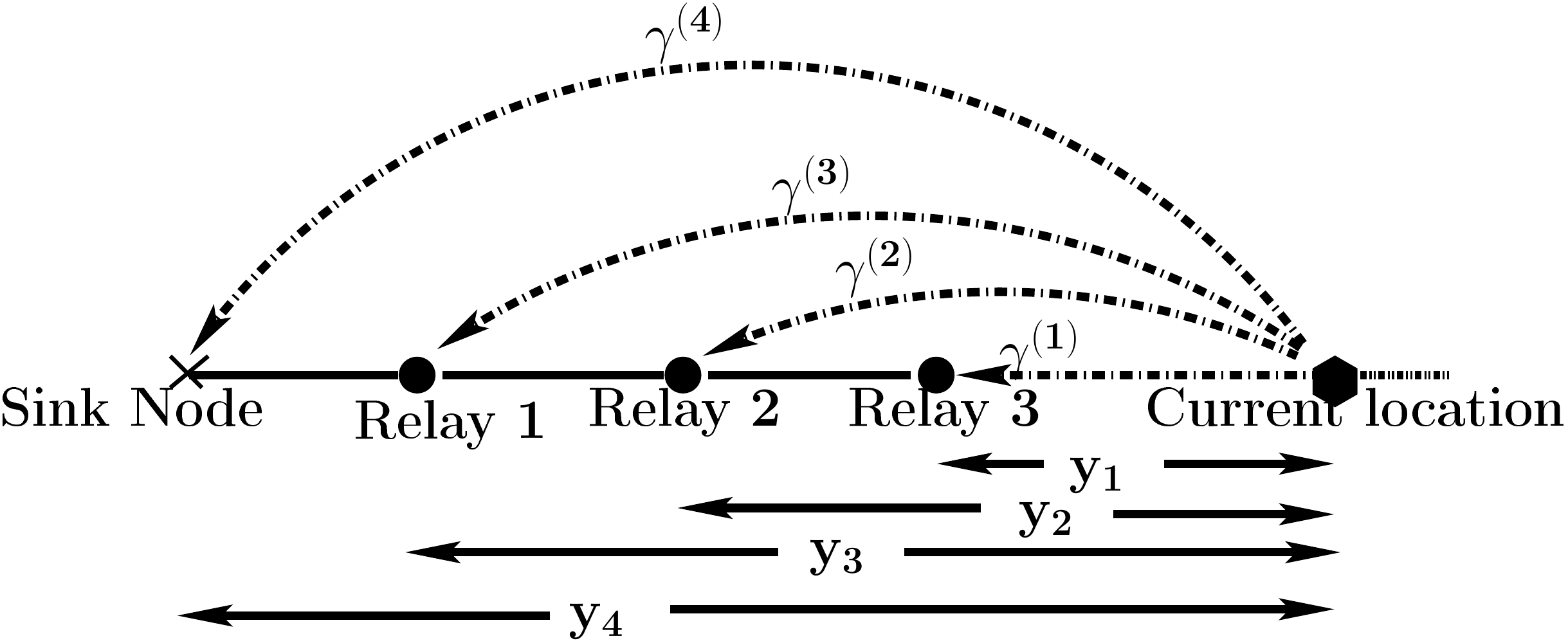}
\vspace{-2mm}
\caption{Measurement-based sequential relay placement. When standing at the 
``current location," the deployment operative, having already deployed Relays $1$, $2$, and $3$,  
makes the power measurements $\gamma^{(k)}$, and knows the distances $y_k$.}
\label{fig:measurement-based-relay-placement}
\vspace{-7mm}
\end{figure}

\begin{figure}[!t]
\centering
\includegraphics[scale=0.32]{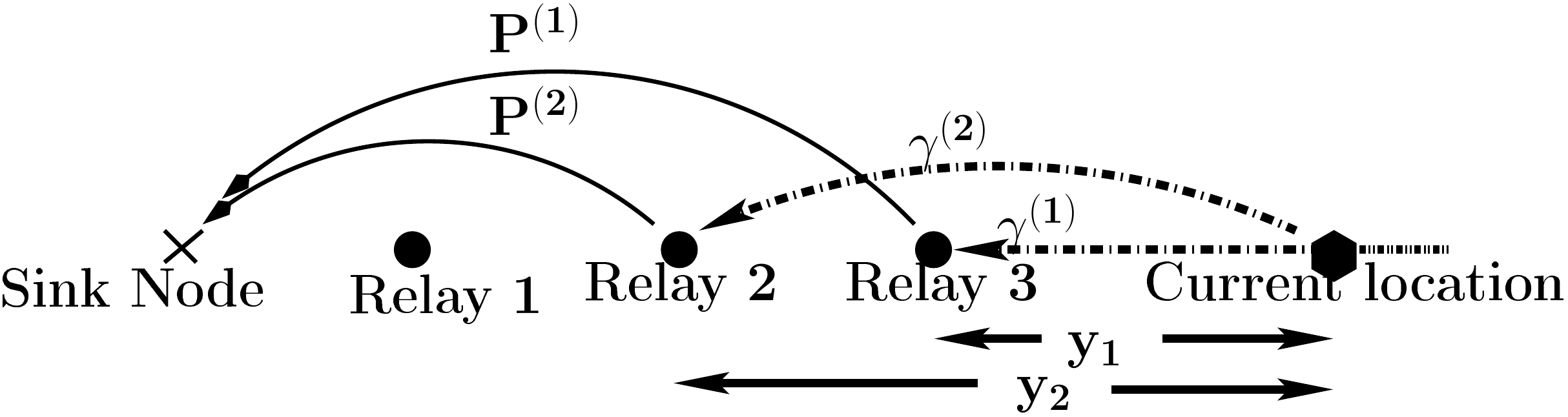}
\vspace{-2mm}
\caption{Sequential deployment of relay nodes with $n=2$. When standing at the
``current location,'' the deployment operative, having already deployed Relays $1$,
$2$, and $3$, makes the power measurements $\gamma^{(1)}$, $\gamma^{(2)}$, and knows the distances $y_1$, $y_2$ and 
the lengths $P^{(1)}$, $P^{(2)}$ of the shortest paths from relay $3$ and relay $2$ to the sink node.}
\label{fig:shortest-path-measurement}
\vspace{-4mm}
\end{figure}

If the state is $(\{y_k\}_{k=1}^{l_m};\{P^{(k)}\}_{k=1}^{l_m};\{\gamma^{(k)}\}_{k=1}^{l_m})$
and a relay is placed, the relay cost $\xi$ is incurred. 
The power cost is incurred only after the source is placed, and that cost will be the length of the 
shortest path in $\mathcal{P}_n$ from the source to the sink. 
Let us define $J_{\xi}(\{y_k\}_{k=1}^{l_m};\{P^{(k)}\}_{k=1}^{l_m};\{\gamma^{(k)}\}_{k=1}^{l_m})$ 
and $J_{\xi}(\mathbf{0})$ to be the
optimal expected cost-to-go starting from state $(\{y_k\}_{k=1}^{l_m};\{P^{(k)}\}_{k=1}^{l_m};\{\gamma^{(k)}\}_{k=1}^{l_m})$ and 
state $\mathbf{0}$ respectively.

\subsubsection{Bellman Equation}\label{subsubsec:bellman_sum_power_shortest_path}
Note that here again we have an infinite horizon total cost MDP with a countable state space, finite action space and 
nonnegative single-stage cost. Hence, 
the optimal value function $J_{\xi}(\cdot)$ satisfies the Bellman equations 
(\ref{eqn:bellman_sum_power_shortest_path}) for $m \geq n$ 
and (\ref{eqn:bellman_sum_power_shortest_path_less_than_memory}) for $m<n$, for the optimal 
cost function. The first term in the $\min \{\cdot,\cdot\}$ is the cost if we place a relay at the state 
$( \{y_k\}_{k=1}^{l_m}; \{P^{(k)}\}_{k=1}^{l_m}; \{\gamma^{(k)}\}_{k=1}^{l_m}  )$, and the second term is the cost if we do not place a relay. 

Observe that it is never optimal to place a relay at state $\mathbf{0}$ because, 
in doing so, a cost $\xi$ will unnecessarily be incurred. Hence, 
$J_{\xi}(\mathbf{0})=\theta \mathbb{E}(\Gamma_1)+(1-\theta)\mathbb{E}J_{\xi}(1;0;\Gamma_1)$.

When $m \geq n$, if we place a relay at the current location and if the line ends in the next step, the length of the shortest path from the source 
to the sink will be seen as a terminal cost, and is equal to 
$\mathbb{E} \min \{\min_{k \in \{1,\cdots,n-1\}} (\Gamma_{y_k+1}+P^{(k)}), \Gamma_{1}+ \min_{k \in \{1,\cdots,n\}} (\gamma^{(k)}+P^{(k)}) \}$. 
Note that in this case the shortest path from the source to the sink can pass via the relay placed at the ``current location'', 
or via one of the $(n-1)$ previous relays. For example, in the scenario shown in Figure~\ref{fig:shortest-path-measurement} (with $n=2$), 
if we place a relay at the ``current location'' and the line ends at the next step, then the neighbouring node 
of the source along the shortest path can be the relay placed at the ``current location'' or relay~$3$ (source is not 
allowed to transmit directly to relay~$2$ because $n=2$).
 Keeping this in mind, $\Gamma_{y_k+1}+P^{(k)}$ is the sum of two costs: the (random) power $\Gamma_{y_k+1}$ 
from the source to the $(k+1)$-st previous node w.r.t the source 
(after placing the relay at the current location, the current $k$-th previous node 
will become the $(k+1)$-st previous node in the next step, where the source will be placed) 
and the length of the shortest path $P^{(k)}$ 
from that node to the sink. 
$\Gamma_{1}+ \min_{k \in \{1,\cdots,n\}} (\gamma^{(k)}+P^{(k)})$ is the sum of the random power $\Gamma_{1}$ 
required to establish a link from the source 
to the relay deployed at the current location, and the length of the shortest path from this relay to the sink.

When $m \geq n$, if we place a relay at the current location and the line does not end in 
the next step, the terms $y_n$, $P^{(n)}$ and $\gamma^{(n)}$ disappear from the state (because a new relay has been 
placed, which must be taken into account in the state) and the distance $1$ of the 
next location from the newly placed relay at the current location is absorbed into the state. 
Other distances in the state increase by $1$ each. 
The length of the shortest path from the newly placed relay to the sink, 
i.e., $\min_{k \in \{1,\cdots,n\}} (\gamma^{(k)}+P^{(k)})$ enters the state, 
and the power required at the next location to connect to the $n$ previous relays (w.r.t the next location) are 
independently sampled again. Hence, keeping in mind that $\Gamma_{r}$ is the random power required 
to establish a link between two nodes at a distance $r$, the new state becomes:

\begin{eqnarray*}
 & & (1,y_1+1,\cdots,y_{n-1}+1;\min_{k \in \{1,\cdots,n\}} (\gamma^{(k)}+P^{(k)}), \\
 & & P^{(1)},\cdots,P^{(n-1)};\Gamma_{1}, \Gamma_{y_1+1},\cdots,\Gamma_{y_{n-1}+1})
\end{eqnarray*}

Similarly, if $m \geq n$ and we do not place a relay at the current location, 
in the next step the line may end with probability $\theta$ and may not end with 
probability $(1-\theta)$. If the line ends, a cost of the shortest path $\mathbb{E} \min_{k \in \{1,\cdots,n\}} (\Gamma_{y_k+1}+P^{(k)})$ 
will be incurred. If the line does not 
end, the next state will be the random tuple $( \{y_k+1\}_{k=1}^n; \{P^{(k)}\}_{k=1}^n; \{\gamma_k\}_{k=1}^n )$, where 
for each $k \in \{1,\cdots,n\}$, $\gamma_k$ will be drawn independently from each other 
from the distribution $G_{y_k+1}(\cdot)$.\footnote{It is to be 
noted that all the $\Gamma_{\cdot}$ terms appearing in (\ref{eqn:bellman_sum_power_shortest_path}), 
 (\ref{eqn:bellman_sum_power_shortest_path_less_than_memory}), (\ref{eqn:bellman_max_power_shortest_path}) 
and (\ref{eqn:bellman_max_power_shortest_path_less_than_memory}) are independent of each other.} 
 
Similar arguments can be used to explain (\ref{eqn:bellman_sum_power_shortest_path_less_than_memory}) in case $m<n$. The 
difference is that if we place a relay at the current location and the line does not end in the next step, 
the next state will have three more terms, since the information for the newly 
placed relay can be accomodated into the state. On the other hand, if the line ends in the next step, the source 
will be able to communicate to the sink via one of the $m$ relays (there will be $(m+1)$ preceding nodes, including the sink).

%

\subsubsection{Results and Discussion}\label{subsubsection:results_discussion_sum_power_shortest_path}

\begin{thm}{\em Policy Structure:}\label{theorem:policy_structure_sum_power_shortest_path}
For the state $( \{y_k\}_{k=1}^{l_m}; \{P^{(k)}\}_{k=1}^{l_m}; \{\gamma^{(k)}\}_{k=1}^{l_m}  )$, the optimal relay placement policy 
is the following:

Place a relay if and only if $\min_{k \in \{1,\cdots,l_m\}} (\gamma^{(k)}+P^{(k)}) \leq c(\{y_k\}_{k=1}^{l_m};\{P^{(k)}\}_{k=1}^{l_m})$ 
where $c(\{y_k\}_{k=1}^{l_m};\{P^{(k)}\}_{k=1}^{l_m})$ is a threshold value.
\end{thm}

\begin{proof}
 See Appendix~\ref{appendix:not_only_adjacent_nodes}.
\end{proof}

{\em Discussion of the Policy Structure:} The structure of the optimal policy as 
stated in Theorem \ref{theorem:policy_structure_sum_power_shortest_path}
is intuitive because here we need to check whether the quantity $\min_{k \in \{1,\cdots,l_m\}} (\gamma^{(k)}+P^{(k)})$ 
which is the length of the shortest path from the current location of the 
deployment operative to the sink, is below a certain threshold. 

{\em \textbf{Remarks:}}
\begin{itemize}
 \item  The optimal cost $J_{\xi}(\mathbf{0})$ of (\ref{eqn:unconstrained_sum_power_problem_shortest}) 
is always less than or equal to that of 
(\ref{eqn:unconstrained_total_power_problem}), if the relay price $\xi$ is same in both cases. This is because 
each policy for $n=1$ will be a policy for $n=2$ as well.
\item $n=\infty$ provides the best policy since there we 
consider information from all previous nodes. 
\end{itemize}

\textbf{Observation:} With $n=1$, $r:=y_1$ and $\gamma:=\gamma^{(1)}$, the Bellman equation 
(\ref{eqn:bellman_sum_power_shortest_path}) reduces to:

\footnotesize
\begin{eqnarray}
 J_{\xi}\bigg(r; P^{(1)}; \gamma \bigg)
&=& \min \bigg\{ \xi+ \theta \mathbb{E} (\Gamma_{1}+ \gamma+P^{(1)}) + \nonumber\\
&&(1-\theta) \mathbb{E} J_{\xi} \bigg(1;(\gamma+P^{(1)}); \Gamma_{1}\bigg), \nonumber\\
 &&\theta \mathbb{E} \bigg(\Gamma_{r+1} + P^{(1)}\bigg)\nonumber\\
&+&(1-\theta) \mathbb{E} J_{\xi} \bigg( r+1; P^{(1)}; \Gamma_{r+1} \bigg) \bigg\}\label{eqn:bellman_sum_power_shortest_path_from_n_one}
\end{eqnarray}
\normalsize

Note that $J_{\xi}(r; P^{(1)}; \gamma )=P^{(1)}+J_{\xi}(r; 0; \gamma )$. Let us denote 
$J_{\xi}(r; 0; \gamma ):=J_{\xi}(r, \gamma )$. Now we can rewrite 
(\ref{eqn:bellman_sum_power_shortest_path_from_n_one}) as:

\footnotesize
\begin{eqnarray}
& & P^{(1)}+J_{\xi}\bigg(r, \gamma \bigg)
= P^{(1)}+\min \bigg\{ \xi+ \gamma+ \theta \mathbb{E} (\Gamma_{1})+ \mathbb{E} J_{\xi} \bigg(1, \Gamma_{1}\bigg), \nonumber\\
&& \theta \mathbb{E} \bigg(\Gamma_{r+1}\bigg)
+(1-\theta) \mathbb{E} J_{\xi} \bigg( r+1,\Gamma_{r+1} \bigg) \bigg\}
\end{eqnarray}
\normalsize

Thus, we obtain the Bellman equation (\ref{eqn:bellman_sum_power}).

\subsection{Max-Power Objective}
Here we are going to address the following problem:
\begin{eqnarray}
 \min_{\pi \in \Pi} \mathbb{E}_{\pi} \bigg( \min_{\mathbf{p} \in \mathcal{P}_n} \max_{e \in \mathbf{E}(\mathbf{p})} \Gamma^{(e)}+\xi N \bigg) 
\label{eqn:shortest_path_max_power_problem}
\end{eqnarray}

We call $\max_{e \in \mathbf{E}(\mathbf{p})} \Gamma^{(e)}$ the ``length" of the path $\mathbf{p}$, and 
$\min_{\mathbf{p} \in \mathcal{P}_n} \max_{e \in \mathbf{E}(\mathbf{p})} \Gamma^{(e)}$ the length of the 
``shortest path" from the source to the sink.
Using notation and arguments similar to those used in problem 
(\ref{eqn:unconstrained_sum_power_problem_shortest}), we can write the Bellman equations 
(\ref{eqn:bellman_max_power_shortest_path}) and (\ref{eqn:bellman_max_power_shortest_path_less_than_memory}) 
and derive the structure of the optimal node placement policy:

\begin{thm}{\em Policy Structure:}\label{theorem:policy_structure_max_power_shortest_path}
For the state $( \{y_k\}_{k=1}^{l_m}; \{P^{(k)}\}_{k=1}^{l_m}; \{\gamma^{(k)}\}_{k=1}^{l_m}  )$, the optimal relay placement policy 
is the following:

Place a relay if and only if $\min_{k \in \{1,\cdots,l_m\}} \max \{\gamma^{(k)},P^{(k)}\} \leq c(\{y_k\}_{k=1}^{l_m};\{P^{(k)}\}_{k=1}^{l_m})$ 
where $c(\{y_k\}_{k=1}^{l_m};\{P^{(k)}\}_{k=1}^{l_m})$ is a threshold value.
\end{thm}

\textbf{Remark:} Note that the Bellman equation (\ref{eqn:bellman_max_power}) can be derived from 
(\ref{eqn:bellman_max_power_shortest_path}), with $n=1$, $r:=y_1$, $\gamma:=\gamma^{(1)}$ and $P^{(1)}=\gamma_{max}$.

\subsection{Performance comparison between $n=1$ and $n=2$}

We have made a comparative study of the performance of the optimal policies with memory $1$ and the policies with memory $2$. 
The results are shown in Table~\ref{table:comparison_memory_one_vs_two}. 
Here we have used the same model as used in Section~\ref{sec:only_adjacent_nodes} in the max-power case, {\em but we have 
considered $\mathcal{S}=\{0.1,0.2,\cdots,2\}$ mW in the sum-power case in order to avoid huge computational 
requirement\footnote{If the transmit power levels in mW are integer multiples of some basic power level, the lengths of the 
shortest paths will also be integer multiples of that basic power level. If the transmit power levels do not satisfy this property, 
the number of possible shortest paths can be very large, leading to enormous computational complexity. This case will not arise 
in the max-power case since, in that case, a shortest path will always take its values from the set $\mathcal{S}$.}.} 
The study suggests that that, for small relay cost, $n=2$ can provide a significant percentage gain over the optimal cost for $n=1$.  
Since at small $\xi$ we tend to place more relays 
(but the relay cost is small compared to $J_{\xi}(\mathbf{0})$, see Table~\ref{table:sum_power_adjacent} and 
Table~\ref{table:max_power_adjacent}), 
skipping relays could be useful. For large $\xi$, we place 
very few relays, but the relay cost will dominate. As $\xi$ becomes very high, we will always place the relays periodically 
at every $10$ steps, and nowhere else; hence the relay cost becomes independent of $n$. 
The little variation in power cost will be insignificant 
compared to large amount of relay cost.

\begin{table}[t!]
\centering
\begin{tabular}{|c |c |c |c|c|}
\hline
 & $\xi=0.001$ & $\xi=0.01$ & $\xi=0.1$ & $\xi=1$\\ \hline
Sum-power  & 0.61946  &  0.65759  & 1.0414   & 4.49396\\ 
($n=1$) & & & & \\ \hline
Sum-power  &  0.50723 &  0.56834  &  1.0233  & 4.43836\\ 
($n=2$) & & & & \\ \hline
Max-power  & 0.03336  &  0.13124   & 0.61693   & 4.10798\\ 
($n=1$) & & & & \\ \hline
Max-power  &   &    &    & \\ 
($n=2$) &0.02119 &0.10686 & 0.60548 & 4.09718\\ \hline
\end{tabular}
\caption{Comparison of $J_{\xi}(\mathbf{0})$ for $n=1$ and $n=2$ for different $\xi$: Sum-power objective and max-power objective. 
Two different transmit power sets are used for the two objectives.}
\vspace{-0.5cm}
\label{table:comparison_memory_one_vs_two}
\end{table}

\subsection{Computational Issues}\label{computational_issues}
The dimension of the state space is $3n$ (increasing in $n$) in the value iteration, 
and hence the computational complexity and memory requirement increases with $n$. 
However, for any arbitrary $n$, we can reduce the value iteration to a function iteration in the same way as in 
(\ref{eqn:new_iteration_sum_power}) and (\ref{eqn:function_iteration_max_power}), and reduce the dimension 
of the domain of the function to $2n$ (instead of $3n$ in the value iteration).

\section{Conclusion}\label{sec:conclusion}

In this paper, we explored several sequential relay placement
problems for as-you-go deployment of wireless relay networks, assuming very light traffic. 
The problems were formulated as MDPs, optimal policies were derived, and the procedure illustrated via numerical examples. 
 There are numerous issues to improve upon: 
(i) the light traffic (``lone packet model") assumption, (ii) the assumption of independent shadow fading from link to link, and 
(iii) the deployment failure issue. 
Extension to positive traffic might require a different approach: perhaps one that requires a performance analysis model 
working in conjunction with an optimal sequential decision technique. We are addressing these issues in our ongoing work.

\renewcommand{\thesubsection}{\Alph{subsection}}

\appendices

\section{Only Links between adjacent nodes permitted}\label{appendix:only_adjacent_nodes}

\textbf{Proof of Lemma \ref{lemma:convergence_value_iteration_sum_power}:} 
Here we have an infinite horizon total cost MDP with countable state space and finite action space. The assumption P of 
Chapter $3$ in \cite{bertsekas07dynamic-programming-optimal-control-2} is satisfied since the single-stage cost is nonnegative. 
Hence, by combining Proposition $3.1.5$ and Proposition $3.1.6$ of \cite{bertsekas07dynamic-programming-optimal-control-2}, 
we obtain the result.

\textbf{Proof of Lemma \ref{lem:properties_value_function_sum_power}:}
In value iteration (\ref{eqn:value_iteration_sum_power}), 
$J_{\xi}^{(0)}(r,\gamma):=0$ is concave, increasing in $\gamma$, $\xi$ and increasing in $r$ and 
$J_{\xi}^{(0)}(\mathbf{0}):=0$ is concave, increasing in $\xi$. Suppose that 
 $J_{\xi}^{(k)}(r, \gamma)$ is concave, increasing in $\gamma$, $\xi$ and increasing in $r$ and 
$J_{\xi}^{(k)}(\mathbf{0})$ is concave, increasing in $\xi$ for some $k \geq 0$. Note that 
$\mathbb{E}(\Gamma_{r+1})$ is increasing in $r$. 

Let us consider $r_1>r_2$. We can write:

\begin{eqnarray*}
&& \mathbb{E}J_{\xi}^{(k)}(r_1+1, \Gamma_{r_1+1})\\
&=& \sum_{\gamma}g(r_1+1,\gamma)J_{\xi}^{(k)}(r_1+1, \gamma)\\
&\geq& \sum_{\gamma}g(r_1+1,\gamma)J_{\xi}^{(k)}(r_2+1, \gamma)\\
&\geq& \sum_{\gamma}g(r_2+1,\gamma)J_{\xi}^{(k)}(r_2+1, \gamma)\\
&=& \mathbb{E}J_{\xi}^{(k)}(r_2+1, \Gamma_{r_2+1})
\end{eqnarray*}
where the first inequality follows from the fact that $J_{\xi}^{(k)}(r+1, \gamma)$ is increasing in $r$ for each $\gamma$, and 
the second inequality follows from the facts that $J_{\xi}^{(k)}(r+1, \gamma)$ is increasing in $\gamma$ and $G_r (\cdot)$ is 
stochastically increasing in $r$. Hence, $\mathbb{E}J_{\xi}^{(k)}(r+1, \Gamma_{r+1})$ 
is increasing in $r$. Hence, by (\ref{eqn:value_iteration_sum_power}), $J_{\xi}^{(k+1)}(r, \gamma)$ 
is increasing in $r$.

We know that the minimum of two concave, increasing functions is concave, increasing. 
Note that, each term in the $\min \{\cdot,\cdot\}$ of (\ref{eqn:value_iteration_sum_power}) is concave, increasing in $\gamma$, $\xi$. 
Hence, $J_{\xi}^{(k+1)}(r, \gamma)$ in concave, increasing in $\xi$, $\gamma$ and $J_{\xi}^{(k+1)}(\mathbf{0})$ is concave, increasing 
in $\xi$. Now, since $J_{\xi}^{(k)}(r,\gamma) \uparrow J_{\xi}(r,\gamma)$ for each $r$, $\gamma$ and 
$J_{\xi}^{(k)}(\mathbf{0}) \uparrow J_{\xi}(\mathbf{0})$ (by Lemma \ref{lemma:convergence_value_iteration_sum_power}), 
the results follow.

\textbf{Proof of Theorem \ref{theorem:policy_structure_sum_power}:}
By Proposition $3.1.3$ of \cite{bertsekas07dynamic-programming-optimal-control-2}, if there exists a stationary policy 
$\{ \mu,\mu,\cdots \}$ such that for each state $(r,\gamma)$, the action chosen by the policy is 
the action that achieves the minimum 
in the Bellman equation (\ref{eqn:bellman_sum_power}), then that stationary policy will be an optimal policy. Hence, 
it is clear that when the state is $(r,\gamma)$, 
it is optimal to place the relay if $c_p \leq c_{np}$, i.e.,
\begin{eqnarray*}
\xi+ \gamma + J_{\xi}(\mathbf{0}) \leq \theta \mathbb{E} 
(\Gamma_{r+1}) + (1-\theta) \mathbb{E} J_{\xi}\left(r+1, \Gamma_{r+1}\right)\\
\end{eqnarray*}
or,
\begin{eqnarray*}
 \gamma  \leq \theta \mathbb{E} 
(\Gamma_{r+1}) + (1-\theta) \mathbb{E} J_{\xi}\left(r+1, \Gamma_{r+1}\right)-(\xi+ J_{\xi}(\mathbf{0}))\label{eqn:placement_condition_sum_power}
\end{eqnarray*}

Thus, the condition for 
placing a relay when the state in $(r,\gamma)$ becomes $\gamma \leq \gamma_{\textit{th}}(r)$, where $\gamma_{\textit{th}}(r)$ is a 
threshold value. Now, by stochastic monotonicity of $\Gamma_{r}$ in $r$, $\mathbb{E} (\Gamma_{r+1})$ is increasing in $r$. 
Also, since $J_{\xi}(r, \gamma)$ is increasing in $r$, $\gamma$ and $\Gamma_{r}$ is stochastically increasing in $r$,
 $\mathbb{E} J_{\xi}\left(r+1, \Gamma_{r+1}\right)$ also is increasing in $r$. Hence, 
$\gamma_{\textit{th}}(r)$ is increasing in $r$.

\textbf{Proof of Lemma \ref{lemma:convergence_value_iteration_max_power}:}
Here we have an infinite horizon total cost MDP with countable state space and finite action space. The assumption P of 
Chapter $3$ in \cite{bertsekas07dynamic-programming-optimal-control-2} is satisfied since the single-stage cost is nonnegative. 
Hence, by combining Proposition $3.1.5$ and Proposition $3.1.6$ of \cite{bertsekas07dynamic-programming-optimal-control-2}, 
we obtain the result.

\textbf{Proof of Lemma \ref{lem:properties_value_function_max_power}:}
Note that, in the value iteration (\ref{eqn:value_iteration_max_power}), 
$J_{\xi}^{(0)}(r,\gamma,\gamma_{max}):=0$ is increasing in $r$, $\gamma$ and $\gamma_{max}$ and concave, increasing 
in $\xi$. Suppose that for some $k \geq 0$, 
$J_{\xi}^{(k)}(r,\gamma,\gamma_{max})$ is increasing in $r$, $\gamma$, $\gamma_{max}$ and 
concave, increasing in $\xi$. Since $J_{\xi}^{(k)}(r,\gamma,\gamma_{max})$ is increasing in $r$, $\gamma$ and $G_r(\cdot)$ is 
stochastically increasing in $r$, $\mathbb{E}J_{\xi}^{(k)}(r+1,\Gamma_{r+1},\gamma_{max})$ is increasing in $r$. 
$\mathbb{E}J_{\xi}^{(k)}(r+1,\Gamma_{r+1},\gamma_{max})$ is also increasing in $\gamma_{max}$, since 
$J_{\xi}^{(k)}(r,\gamma,\gamma_{max})$ is increasing in $\gamma_{max}$. $\mathbb{E}(\gamma_{max},\Gamma_{r+1})$ is increasing 
in $\gamma_{max}$ and also increasing in $r$ since $G_r$ is stochastically increasing in $r$. On the other hand, the 
first term in the $\min\{\cdot,\cdot\}$ of (\ref{eqn:value_iteration_max_power}) is independent of $r$, but increasing in 
$\gamma$, $\gamma_{max}$ and linearly increasing in $\xi$. Now, minimum of two increasing functions is increasing and 
the minimum of a linear function and a constant is concave. Hence, $J_{\xi}^{(k+1)}(r,\gamma,\gamma_{max})$ 
is increasing in $r$, $\gamma$, $\gamma_{max}$ and concave, increasing in $\xi$. 
Since $J_{\xi}^{(k)}(r,\gamma,\gamma_{max}) \uparrow J_{\xi}(r,\gamma,\gamma_{max})$, the results follow.

\textbf{Proof of Theorem \ref{theorem:policy_structure_max_power}:}
By similar arguments as used in the proof of Theorem \ref{theorem:policy_structure_sum_power}, the condition for placing a 
relay at a state $(r,\gamma,\gamma_{\textit{max}})$ is:

 \footnotesize
\begin{eqnarray}
 \xi+ \theta \mathbb{E} \max\{\gamma, \gamma_{\textit{max}}, \Gamma_{1}\}
+ (1-\theta) \mathbb{E}J_{\xi}(1,\Gamma_{1}, \max\{\gamma, \gamma_{\textit{max}}\})\nonumber \\
\leq 
\theta \mathbb{E} \max \{\gamma_{\textit{max}}, \Gamma_{r+1}\}
+ (1-\theta) \mathbb{E}J_{\xi}(r+1, \Gamma_{r+1}, \gamma_{\textit{max}}) \label{eqn:placement_condition_max_power}
\end{eqnarray}
 \normalsize
Note that $\mathbb{E} \max \{\gamma_{\textit{max}}, \Gamma_{r+1}\}$ increases in $r$. Also, 
$J_{\xi}(r+1, \gamma, \gamma_{\textit{max}})$ is increasing in $r$ and $\gamma$ for all $\gamma_{\textit{max}}$. 
Hence, $\mathbb{E}J_{\xi}(r+1, \Gamma_{r+1}, \gamma_{\textit{max}})$ is increasing in $r$, since $G_{r}$ is 
stochastically increasing in $r$. Hence, the R.H.S of (\ref{eqn:placement_condition_max_power}) 
is increasing in $r$ and $\gamma_{\textit{max}}$. 
Now, if $\gamma \leq \gamma_{\textit{max}}$, the L.H.S of (\ref{eqn:placement_condition_max_power}) is independent of $\gamma$. 
Hence, the condition for placing the relay is $r \geq r_{\textit{th}}(\gamma_{\textit{max}})$ where 
$r_{\textit{th}}(\gamma_{\textit{max}})$ is a threshold value.
On the other hand, if $\gamma > \gamma_{\textit{max}}$, the L.H.S of (\ref{eqn:placement_condition_max_power}) 
is independent of $\gamma_{\textit{max}}$ and increasing in $\gamma$. 
Hence, we will place the relay if and only if  $\gamma \leq \gamma_{\textit{th}}(r, \gamma_{\textit{max}})$ 
where $\gamma_{\textit{th}}(r, \gamma_{\textit{max}})$ is 
a threshold value increasing in $r$ and $\gamma_{\textit{max}}$.

\section{Links between any pair of nodes permitted}\label{appendix:not_only_adjacent_nodes}

\textbf{Proof of Theorem \ref{theorem:policy_structure_sum_power_shortest_path}:}
By similar arguments using convergence of the value iteration as in the proof of Lemma~\ref{lem:properties_value_function_sum_power}, 
we can claim that 
$J_{\xi}( \{y_k\}_{k=1}^{l_m}; \{P^{(k)}\}_{k=1}^{l_m}; \{\gamma^{(k)}\}_{k=1}^{l_m}  )$ is increasing in each of its arguments. 
Now, the condition for placing a relay is that the first term in the $\min \{\cdot,\cdot\}$ of 
the R.H.S of (\ref{eqn:bellman_sum_power_shortest_path}) or (\ref{eqn:bellman_sum_power_shortest_path_less_than_memory}) 
is less than or 
equal to the second term. Since the first term is increasing in 
$\min_{k \in \{1,\cdots,l_m\}} (\gamma^{(k)}+P^{(k)})$, the threshold nature of the optimal relay placement policy is evident.

\bibliographystyle{unsrt}
\bibliography{arpan-techreport}

\end{document}